%% file: paperv.tex
\documentclass[journal]{IEEEtran}
\usepackage{cite}
\usepackage{amsmath,amssymb,amsfonts}
\usepackage{amsthm}
\usepackage{subcaption}
\usepackage{textcomp}
\usepackage{hyperref}
\usepackage{nicefrac}
\usepackage{soul}
\usepackage{pifont}
\usepackage{xcolor}
\usepackage{cuted}
\usepackage[pdftex]{graphicx}
\usepackage{algorithm}
\usepackage{algpseudocode}
\usepackage{ifthen}
\usepackage{graphicx}
\usepackage{mathtools}
\usepackage{diagbox}

\algnewcommand\algorithmicinput{\textbf{input:}}
\algnewcommand\algorithmicoutput{\textbf{output:}}
\algnewcommand\algorithmicbl{\textbf{\# of blocks:}}
\algnewcommand\algorithmicth{\textbf{\# of threads / block:}}
\algnewcommand\Input{\item[\algorithmicinput]}
\algnewcommand\Output{\item[\algorithmicoutput]}
\algnewcommand\Block{\item[\algorithmicbl]}
\algnewcommand\Thread{\item[\algorithmicth]}
\algrenewcommand\algorithmicindent{1.0em}

\newcommand{\llr}{\ell}
\newcommand{\LLR}{L} 

\newcommand{\BUTI}{f} 

\newcommand{\TREESYM}{\tau}
\newcommand{\TREEI}[2][]{\TREESYM_{#2\rightarrow}^{#1}}
\newcommand{\TREEF}[2][]{\TREESYM_{\rightarrow#2}^{#1}}

\newcommand{\BO}{\alpha_{out}} 
\newcommand{\BORF}{\hat{\alpha}_{out}}
\newcommand{\BOB}[2]{\alpha_{out,#1}^{#2}} 
\newcommand{\BOBRF}[2]{\hat{\alpha}_{out,#1}^{#2}}
\newcommand{\BOBL}[1]{\theta^{#1}}
\newcommand{\BOBLRF}[1]{\hat{\theta}^{#1}}
\newcommand{\BOBMAT}[1]{\Theta_{#1}} 
\newcommand{\BOBMATRF}[1]{\hat{\Theta}_{#1}}

\newcommand{\BI}{\alpha_{in}} 
\newcommand{\BM}{\delta}
\newcommand{\BMRF}{\hat{\delta}}
\newcommand{\PM}{\lambda}
\newcommand{\DPM}{\dot\lambda}
\newcommand{\BMMAT}{\Delta}
\newcommand{\PMMAT}{\Lambda}
\newcommand{\DPMMAT}{\dot\Lambda}

\newcommand{\s}{s}

\newcommand{\RADIXI}{\rho}

\newcommand{\BITIDX}{:}

\newcommand{\SP}{\phi}
\newcommand{\K}{k} 
\newcommand{\B}{\beta} 
\newcommand{\N}{n} 
\newcommand{\F}{f} 
\newcommand{\Ovr}{v} 

\newcommand{\BG}{DG}

\newtheorem{theorem}{Theorem}
\newtheorem{corollary}{Corollary}[theorem]

\DeclareMathOperator*{\argmax}{arg~max}

\begin{document}
\title{High-Throughput Parallel Viterbi Decoder on GPU Tensor Cores}
%
%
%

\author{Alireza~Mohammadidoost,~
Matin~Hashemi
\thanks{A. Mohammadidoost and M. Hashemi are with the Department of Electrical Engineering, Sharif University of Technology, Tehran, Iran. E-mails: mohammadidoost@ee.sharif.edu, matin@sharif.edu (corresponding author).}}

%

\maketitle

\begin{abstract}
Many research works have been performed on implementation of Vitrerbi decoding algorithm on GPU instead of FPGA because this platform provides considerable flexibility in addition to great performance. Recently, the recently-introduced Tensor cores in modern GPU architectures provide incredible computing capability. This paper proposes a novel parallel implementation of Viterbi decoding algorithm based on Tensor cores in modern GPU architectures. 
The proposed parallel algorithm is optimized to efficiently utilize the computing power of Tensor cores. Experiments show considerable throughput improvements in comparison with previous works.
%
\end{abstract}


%
\IEEEpeerreviewmaketitle

\section{Introduction}
\label{sec:intro}
%
%
%
%

\IEEEPARstart{C}{hannel} coding is a technique that is widely employed in data transmission over an unreliable or noisy communication channel. The transmitter encodes the original message by adding redundancy. This enables the receiver to recover the original data by decoding the received noisy data. 
Convolutional coding is a channel coding method which has been widely used in industrial protocols, for instance, in DVB-T, DVB-S, GPRS, GSM, LTE, 3G, CDMA, WiFi and WiMAX. 
Different decoding algorithms exist for convolutional codes, among which the Viterbi decoding algorithm is the optimal and the most widely-used method \cite{Viterbi}. 

The Viterbi decoder operates in either hard-decision mode or soft-decision mode. In the hard-decision mode, every bit in the input of the decoder is represented by either a zero or one. In the soft-decision mode, however, every input bit is a log likelihood ratio (LLR) that is formed based on the probability that the received bit is zero or one. 
In this mode, the Viterbi decoder takes advantage of the additional information in order to decrease the bit error rate (BER) by about $2$~dB. A lower BER means a better recovery of the original signal. This comes at the cost of higher computational requirement, which in turn, lowers the overall decoding throughput.

Many FPGA-based methods have been proposed for acceleration of the Viterbi decoding algorithm. While such methods achieve very high throughput, they do not provide the flexibility required for software defined radio (SDR) and cognitive radio (CR) applications.

This paper proposes a novel parallel algorithm for implementation of the Viterbi decoding algorithm in the soft-decision mode on GPU hardware. 
Flexibility is a key factor in software defined radio (SDR) and GPU provides a platform for flexible software-based implementations at high throughput.

The proposed solution is mainly focused on optimizing the algorithm to fully employ tensor cores.


\input{prelim}

\input{relwork}
\input{radix2}
\input{tensor2}
\input{radix4}

\input{tensor4}
\input{results}




%




\ifCLASSOPTIONcaptionsoff
  \newpage
\fi

\end{document}

%% file: prelim.tex
\section{Preliminaries}
\label{sec:prelim}

This section presents a brief overview of convolutional encoding, Viterbi decoding algorithm, and the concept of soft-decision inputs. In addition, a brief overview of CUDA API for parallel programming on GPU hardware is presented.


\subsection{Convolutional Encoder} 

Fig.~\ref{fig:encoder}(a) shows an example \textbf{convolutional encoder}. The encoder receives a series of $\N$  bits and produces a series of encoded bits which will be transmitted over the communication channel. 
$\B \geq 2$ encoded bits are generated for every input bit. \textbf{Code rate} is defined as the inverse of $\B$. 
At time (stage) $t$, each one of the $\B$ output bits is computed based on the current input bit, i.e., $in_t$, and the previous $\K-1$ input bits as 
\begin{equation}
(g_{\K-1} . in_{t}) \oplus \cdots \oplus (g_0 . in_{t-\K+1}) 
\label{eq:encoder}
\end{equation}
where, $\K$ is called the \textbf{constraint length}, $\oplus$ is the $xor$ operator, and $g$ is a $\K$-bit value called the \textbf{generator polynomial}. There are $\B$ generator polynomials, one for every output bit. 
In the example shown in Fig.~\ref{fig:encoder}(a), $\K=7$ and $\B=2$. The two generator polynomials are $1111001$ and $1011011$, whose octal representations are $171$ and $133$, respectively.

The encoder can be viewed as a finite state machine (FSM) with $2^{\K-1}$ \textbf{states}. The previous $\K-1$ input bits, i.e., $(in_{t-1}, in_{t-2}, \ldots in_{t-\K+1})$, form the current state. 
Assume the FSM is in state $i \in [0,2^{\K-1}-1]$. Depending on $in_t$, i.e., the current input bit which is either zero or one, the FSM takes a \textbf{branch} $ij$ from state $i$ to state $j$. 
Hence, a series of input bits causes the FSM to take a series of branches, which is called a \textbf{path}.

The encoder FSM can be formed solely based on $\K$, $\B$ and the generator polynomials. Every branch $ij$ in the FSM is associated with a single-bit \textbf{branch input} $\BI^{ij}$ and a $\B$-bit \textbf{branch output} $\BO^{ij}$. See Fig.~\ref{fig:encoder}(b). The output of the encoder is $\BO^{ij}$ provided that the FSM goes from state $i$ to $j$. This transition happens if the input bit is equal to $\BI^{ij}$.

\subsection{Viterbi Decoding Algorithm} 

The encoded bits are transmitted over the channel. When they arrive at the receiver, some of the bits are corrupted. The received bits are decoded in order to recover the original data. 
The Viterbi decoding algorithm is the optimal and the most widely-used method for decoding convolutional codes.

Given the received series of bits, the decoding algorithm recovers the original data by finding the most probable path swept by the encoder FSM. 
This is done by investigating the likelihood of many different paths. 
Since the number of possible paths grow exponentially, the Viterbi decoding algorithm employs a dynamic programming approach to efficiently find the most probable path as the following. 

First, for branch $ij$ from state $i$ to state $j$ at stage $t \in [0,\N)$, \textbf{branch metric} $\BM^{ij}_{t}$ is computed as 
\begin{equation}
\BM^{ij}_t = \sum_{b=0}^{\B-1} (-1)^{\displaystyle \BO^{ij}[b]} \times \llr_t[b]
\label{eq:BM}
\end{equation}
where, $\BO^{ij}[b]$ and $\llr_t[b]$ denote bit $b$ of $\BO^{ij}$ and $\llr_t$, respectively. $b \in [0,\B)$. The term $\llr_t$ denotes the received $\B$ bits at the decoder input at time $t$. Basically, $\BM^{ij}_t$ indicates the amount of similarity between the received data and the output of branch $ij$. 
%
%
%
Next, for state $j$ at stage $t$, \textbf{path metric} $\PM^j_t$ is computed as 
\begin{equation}
\PM^j_t = \max_{i \in \text{prv}(j)} \big( 
\PM^{i}_{t-1} + \BM^{ij}_t  
\big)
\label{eq:PM}
\end{equation}
where, $i$ iterates over the previous states of state $j$ in the FSM. Note that every state has two previous states, i.e., two input branches. 
The above equation is known as ACS (add, compare and select) operation. Basically, $\PM^j_t$ is formed by accumulating a series of branch metrics which eventually end at state $j$ at stage $t$, and also, have the highest possibility. In other words, $\PM^j_t$ indicates the likelihood of the most probable path which ends at state $j$ till stage $t$.

While computing $\PM^j_t$, the selected previous state $\SP^j_t$ which maximized $\PM^j_t$ needs to be saved as well. It basically indicates one stage of the \textbf{survivor path} which ends at state $j$ at stage $t$. 
\begin{equation}
\SP^j_t = \argmax_{i \in \text{prv}(j)} \big( 
\PM^{i}_{t-1} + \BM^{ij}_t  
\big)
\label{eq:SP0}
\end{equation}
%

The above calculations constitute the forward procedure in the Viterbi decoding algorithm. See Alg.~\ref{alg:serialVit1}. The outer loop iterates through all stages $t \in [0, \N)$ and the inner loop iterates through all states $j \in [0, 2^{\K-1})$.

Once the forward procedure is complete, the backward procedure is performed as the following. 
First, the most probable survivor path at the last stage, i.e., $t=\N-1$, is selected as the winner, and that specific path is traced back to the first stage, i.e., $t=0$. At every stage during the trace-back, the decoder output (which is ideally equal to the original data) is formed based on $\BI^{ij}$. The two operations are called \textbf{trace-back} and \textbf{decoding}. See Alg.~\ref{alg:serialVit2}.

\begin{figure}[tp]
	\centering
	\includegraphics[width=1.0\columnwidth]{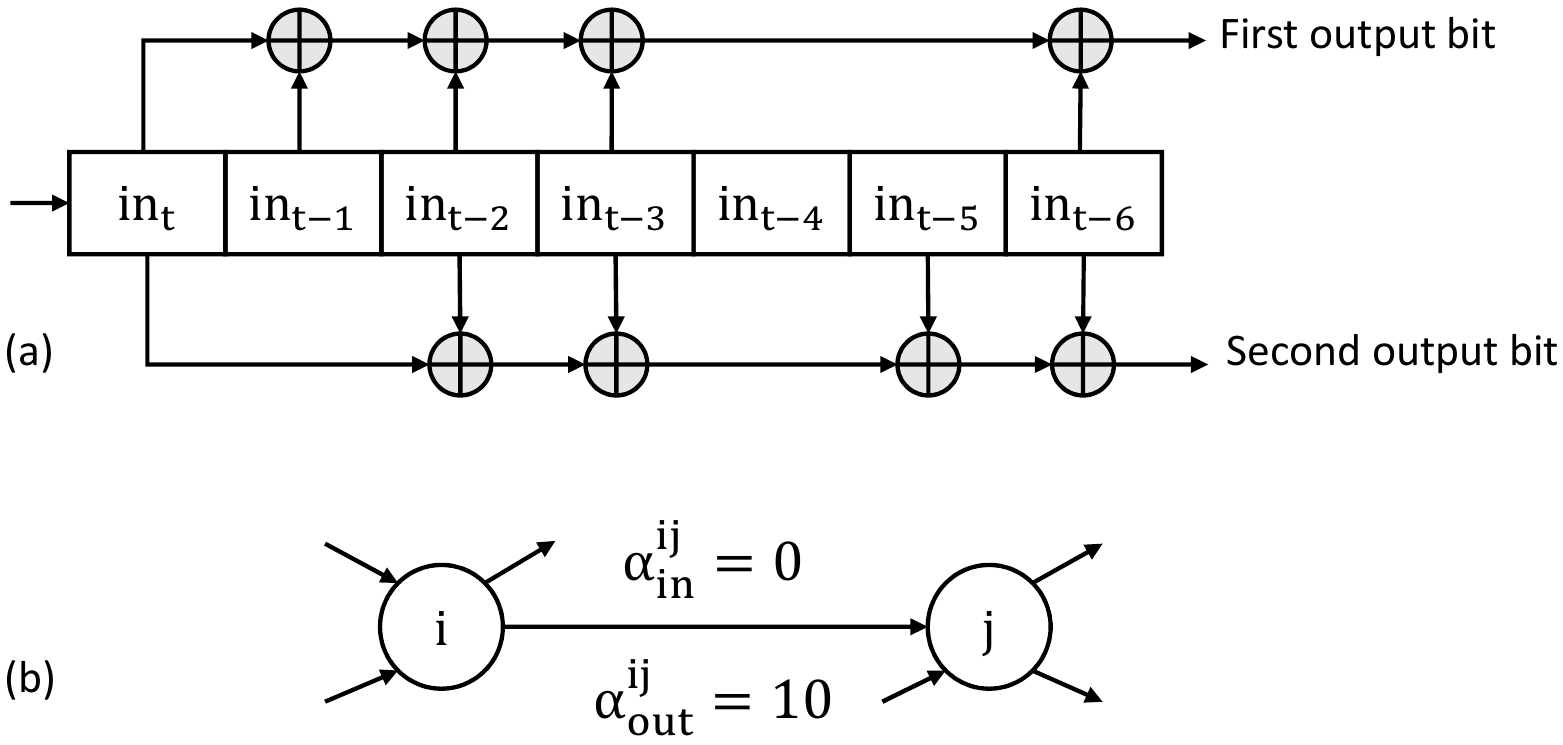} 
	\caption{a) Convolutional encoder $(\B, 1, \K)$ with $\K=7$, $\B=2$, and generator polynomials $1111001$ and $1011011$. b) Branch $ij$ from state $i$ to $j$ in the encoder FSM.}
	\label{fig:encoder}
\end{figure}

\begin{algorithm}[tp]
	\begin{algorithmic}[1]
		\Input $\llr [\B][\N]$ 
		\Output $\SP[2^{\K-1}][\N]$
		
		\State Initialize $\BM[2^{\K-1}][\N]$ and $\PM[2^{\K-1}][\N]$ to zero
		\For{$t=0$ to $\N-1$}
		\For{$j=0$ to $2^{\K-1}-1$}
		\State Set $i'$ and $i''$ equal to the two previous states of $j$
		\State Compute $\BM[i'j,t]$ and $\BM[i''j,t]$ according to \eqref{eq:BM}
		\State $\PM' = \PM[i',t-1] + \BM[i'j,t]$
		\State $\PM'' = \PM[i'',t-1] + \BM[i''j,t]$	
		\If{$\PM' > \PM''$}
		\State $\PM[j,t] = \PM'$
		\State $\SP[j,t] = i'$
		\Else
		\State $\PM[j,t] = \PM''$
		\State $\SP[j,t] = i''$
		\EndIf
		\EndFor
		\EndFor
	\end{algorithmic}
	\caption{\small The first step, i.e., the forward procedure, in the Viterbi decoding method.}
	\label{alg:serialVit1}
\end{algorithm}

\begin{algorithm}[tp]
	\begin{algorithmic}[1]
		\Input $\SP[2^{\K-1}][\N]$
		\Output $Out[\N]$ 
		\State  $j^* = \argmax\limits_{j \in [0,2^{\K-1})} \PM[j,\N-1]$
		\For{$t=\N-1$ to $0$}
		\State $i = \SP[j^*,t]$
		\State $Out[t] = \BI^{ij^*}$ 
		\State $j^*=i$
		\EndFor
	\end{algorithmic}
	\caption{The second step, i.e., the backward procedure, in the Viterbi decoding method.}
	\label{alg:serialVit2}
\end{algorithm}


\subsection{Hard-decision vs. Soft-decision} 

The input to the Viterbi decoding algorithm can be represented in either hard-decision mode or soft-decision mode. In  \textbf{hard-decision} mode, every bit is simply represented by either zero or one. 
In \textbf{soft-decision} mode, however, every bit in the input of the decoder is a \textbf{log likelihood ratio (LLR)} that is formed by the receiver circuit based on the probability that the received bit is zero or one. A larger positive LLR means a larger probability of zero, and a larger negative LLR means a larger probability of one. 
Basically, in addition to the indication of the value of the input bit, the truth of this indication is provided as well.

The Viterbi decoder can take advantage of this additional information in order to better recover the original data. Bit error rate (BER) is lower in the soft-decision mode by about $2$~dB. This comes at the cost of higher computational requirement.

\subsection{CUDA Parallel Programming API}
\label{sec:prelim:cuda}

CUDA is a parallel programming API for Nvidia GPUs. GPU is a massively parallel processor with hundreds to thousands of cores. 
CUDA follows a hierarchical programming model. At the top level, computationally intensive functions are specified by the programmer as CUDA \textbf{kernels}. A kernel is specified as a sequential function for a single \textbf{thread}. The kernel is then launched for parallel execution on the GPU by specifying the number of concurrent threads. 

Threads are grouped into \textbf{blocks}. A kernel consists of a number of blocks, and every block consists of a number of threads. 
%
%
In order to identify blocks within a kernel, and also, threads within a block, a set of indices are used in the CUDA API, for instance,  $blockIdx.x$ as the block index in dimension $x$ within a kernel, and $threadIdx.x$ as the thread index in dimension $x$ within a block. 

%% file: relwork.tex
\section{Previous GPU-Accelerated Viterbi Decoder Methods}
\label{sec:relwork}

The Viterbi decoder algorithm is inherently a sequential procedure and the amount of available parallelism is minimal. 
In specific, calculating the branch metrics is the only step which can be fully parallelized, and path metric calculations can only be partially parallelized using at most $2^{\K-1}$ threads, e.g., $2^{7-1}=64$ threads, where each thread sequentially iterates over $\N$ stages. This approach was proposed in \cite{zhang2009,kim2010ieeecomm}. 
%
%


In order to better utilize the parallel computing capabilities of GPU and increase the throughput, the decoder output can be estimated by dividing the $\N$ stages into many small frames (tiles) of $\F$ stages each \cite{lin2011tiling, gautam2014opencl, lee2013, lee2014}. Every frame is processed in parallel with other frames, and decodes $\F$ out of $\N$ output bits. 
This tiling scheme increases the amount of available parallelism by a factor of $\nicefrac{\N}{\F}$. However, BER is degraded because not all previous history is available in a frame. In order to reduce BER degradation, consecutive frames should have small overlaps in order to carry enough history for correct decoding \cite{lin2011tiling, gautam2014opencl, lee2013, lee2014}. 
As shown in Fig.~\ref{fig:relworks}(b), every frame has an overlap of length $\Ovr$ with its neighbor frames. Hence, to generate $\F$ decoded bits as the output, a frame needs to process $\F+\Ovr$ stages of the original Viterbi decoding algorithm. 


In the forward procedure, the resulting survivor paths need to be stored in GPU global memory for later use in the backward procedure \cite{lin2011tiling,gautam2014opencl,lee2013,lee2014}. The amount of GPU global memory required to store the survivor paths is in the order of 
\begin{equation}
O \Big( 2^{\K-1} \times \N \times (1+\frac{\Ovr}{\F}) \Big)
\end{equation}
because there are $\nicefrac{\N}{\F}$ frames, and every frame requires $O(2^{\K-1} \times (\Ovr+\F))$ space from GPU global memory. 
The parallel algorithms proposed in \cite{li2013,li2014} improved the throughput by judiciously combining the execution of multiple frames in order to coalesce the memory accesses of the survivor paths in global memory. The coalesced memory accesses result in higher throughput.

In \cite{peng2016}, in addition to the tiling scheme and coalescing accesses of survivor paths, branch metrics are efficiently computed according to specific repetitive patterns which help to share computations. In addition, data transfers between CPU and GPU are optimized, in specific, by employing multiple CUDA streams, and by compacting every four input $llr$ values as a $32$-bit value, and every $32$ output decoded bits as a $32$-bit value. 


%


%% file: radix2.tex

\section{butterfly patterns in the trellis}
\label{sec:radix2}

The Viterbi algorithm procedure is represented by a graph called trellis. The trellis has butterfly-like patterns with some features that can be employed to improve the implementation. In this section, such characteristics are explored. 

Based on the shift register shown in Fig. 1(a).... Every state $i$ in stage $t$ in the trellis has two output branches to two separate states $j_0$ and $j_1$ in stage $t+1$. The branches have the following characteristic. For every two successive states $i$ and $i+1$ ($i$ is even) in stage $t$, the four output branches terminate to only two unique states $j_0$ and $j_1$ in stage $t+1$. EXAMPLE. Since every state has only two incoming branches a butterfly pattern is formed. Fig.~\ref{fig:radix2}(a) shows that a butterfly is composed of four branches between two pairs of states in two consecutive stages. Since each butterfly takes two states in a stage, the number of the butterflies is $2^{\K-2}$ which is half of the number of the states as depicted in Fig.~\ref{fig:radix2}(a). In addition, butterflies are isolated sub-graphs constituting a disconnected graph between two successive stages.

Butterflies are indexed from top to bottom and their indexes are denoted by $\BUTI$. Some algorithm parameters like $\BO^{i,j}$ and $\PM_t$ are indexed corresponding to the state and the stage at which they are placed in trellis. In other words, they are indexed based on their position in the trellis. Such indexes that are associated with the position in whole trellis are called global indexes. Global state index and global stage index are two kinds of global indexes that will be referred.

Besides the position in the trellis, the local position in the butterfly can be considered to index parameters of a single butterfly since a butterfly is a 2-state trellis. In this view, local state index and local stage index can be defined. Fig.~\ref{fig:radix2}(b) shows that states of the first local stage have global indexes of $i_0$ and $i_1$ and those of the second local stage have $j_0$ and $j_1$. It is also depicted that other parameters are indexed likewise. Hence, all parameters are indexed based on local indexes and the butterfly index is also mentioned. Since now on, states in the first and second local stages are called left and right states, respectively.

As follows, the relation between local and global state indexes are explored and some other features will be probed.

\begin{theorem}
	\label{theorem:r2:l2g}
	In order to retrieve the value of $i_0$, $i_1$, $j_0$ and $j_1$ in the butterfly $\BUTI$, relations in Eq.~\ref{eq:radix2:S} can be employed.
	\begin{equation}
	\begin{matrix*}[l]
	i_0=2\BUTI & , & j_0=\BUTI\\
	i_1=2\BUTI+1 & , & j_1=\BUTI+2^{\K-2}
	\end{matrix*}
	\label{eq:radix2:S}
	\end{equation}
\end{theorem}

\begin{figure}[h]
	\centering
	\includegraphics[width=0.9\columnwidth]{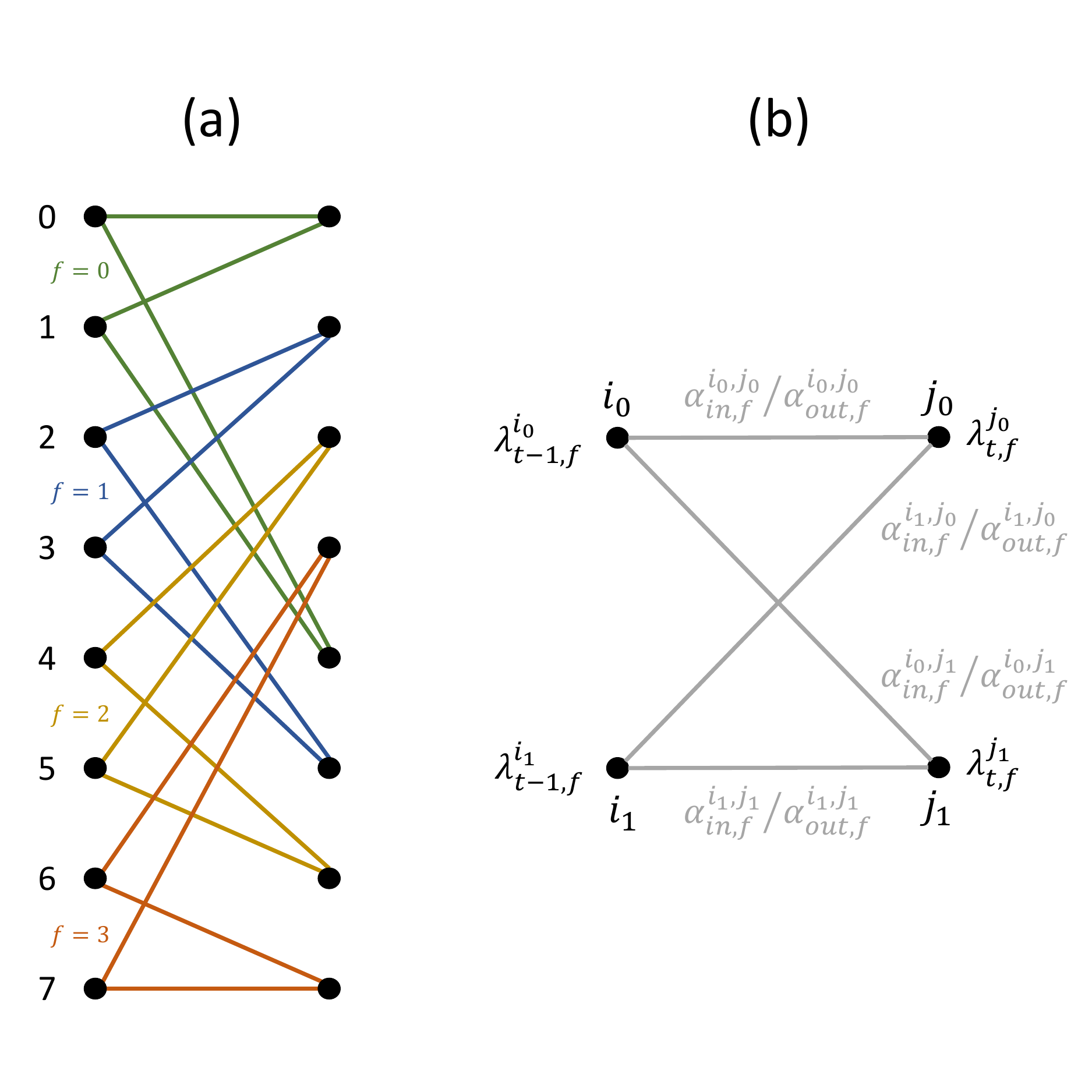} 
	\caption{(a) butterflies in a trellis with 8 states. (b) a single butterfly showing the parameters with butterfly-wise indexing)}
	\label{fig:radix2}
\end{figure}


\begin{proof}
Because the LSB is the bit that will be shifted out in the next stage, its value does not affect the state in the next stage. As a result, the right state is the same regardless of the LSB bit in the left state. Equivalently, two states different in only their LSB bits are left states of a butterfly. That is the reason why left states of a butterfly are two successive states beginning with an even one.

The incoming bit to the shift-register, which is the input bit of the encoder, plays a similar role in destination states of a butterfly. As this incoming bit will be the MSB bit of two right states, right states are different in their MSB bits. Therefore, these two states are $2^{k-2}$ away from each other.
\end{proof}

\begin{theorem}
	All branch outputs of a given butterfly are related to one another. Specifically, the relation between the first branch output and other ones is shown in Eq.~\ref{eq:radix2:outputRelation}.
	\begin{equation}
	\forall i,j<2 \phantom{a} \exists func(x) : \BOBRF{\BUTI}{i,j} = func(\BOBRF{\BUTI}{i_0,j_0})
	\label{eq:radix2:outputRelation}
	\end{equation}
\end{theorem}

\begin{proof}
Based on Eq.~\ref{eq:encoder}, Eq.~\ref{eq:radix2:butterflyOutputBegin} to Eq.~\ref{eq:radix2:butterflyOutputEnd} represent the $b$-th bits of branch outputs of a butterfly, that are associated with the $b$-th generative polynomial $g$, which is a $\K$-bit number. 

\begin{align}
\label{eq:radix2:butterflyOutputBegin}
\BOB{\BUTI}{i_0,j_0}[b] &= (g_{\K-1}.0) \oplus (g_{\K-2}.in_{t-1}) \oplus \cdots \oplus (g_{0}.0) \\
\BOB{\BUTI}{i_0,j_1}[b] &= (g_{\K-1}.1) \oplus (g_{\K-2}.in_{t-1}) \oplus \cdots \oplus (g_{0}.0) \\
\BOB{\BUTI}{i_1,j_0}[b] &= (g_{\K-1}.0) \oplus (g_{\K-2}.in_{t-1}) \oplus \cdots \oplus (g_{0}.1) \\
\BOB{\BUTI}{i_1,j_1}[b] &= (g_{\K-1}.1) \oplus (g_{\K-2}.in_{t-1}) \oplus \cdots \oplus (g_{0}.1)
\label{eq:radix2:butterflyOutputEnd}
\end{align}

Each term concerns a specific input bit from time $t-\K+1$ to $t$. Since state is concatenation of input bits from time $t-\K+1$ to $t-1$, it can be also said that each term concerns either current input bit bit or one bit of state. In each equation, the first term concerns current input bit, which is the MSB bit of the right state, and the last term is related to the LSB bit of state. Regarding the proof of the theorem~\ref{theorem:r2:l2g}, these two bits are known for every branch and they the only differences among states in a butterfly. In other words, the other terms are the same in four equations. In Eq.~\ref{eq:radix2:butterflyOutputBegin}, the first and the last term are both zero. It means that it equals to middle terms that are shared between four equations. So

\begin{align}
\label{eq:radix2:boRelatBegin}
\BOB{\BUTI}{i_0,j_1}[b] &= (g_{\K-1}.1) \oplus \BOB{\BUTI}{i_0,j_0}[b] \oplus (g_{0}.0) \\
\BOB{\BUTI}{i_1,j_0}[b] &= (g_{\K-1}.0) \oplus \BOB{\BUTI}{i_0,j_0}[b] \oplus (g_{0}.1) \\
\BOB{\BUTI}{i_1,j_1}[b] &= (g_{\K-1}.1) \oplus \BOB{\BUTI}{i_0,j_0}[b] \oplus (g_{0}.1)
\label{eq:radix2:boRelatEnd}
\end{align}

\end{proof}

\begin{corollary}
	If the LSB and MSB bits of all generative polynomials are both $1$ (which is the case in most standards such as CCSDS, DVB-S and DVB-T), then
	\begin{equation}
	\BOB{\BUTI}{i_0,j_0} = \BOB{\BUTI}{i_1,j_1} = \overline{\BOB{\BUTI}{i_0,j_1}} = \overline{\BOB{\BUTI}{i_1,j_0}}
	\label{eq:radix2:boRelatNot}
	\end{equation}
\end{corollary}

In other words, two outer branches have the same output and two inner branches also have the same outputs that is toggled version of the outer ones.

In the following, we make use of the above theorems in the proposed solution.

%% file: tensor2.tex
\section{Employing Tensor Cores}
\label{sec:tensorRadixII}

\begin{figure*}[t]
	\centering
	\includegraphics[width=0.8\textwidth]{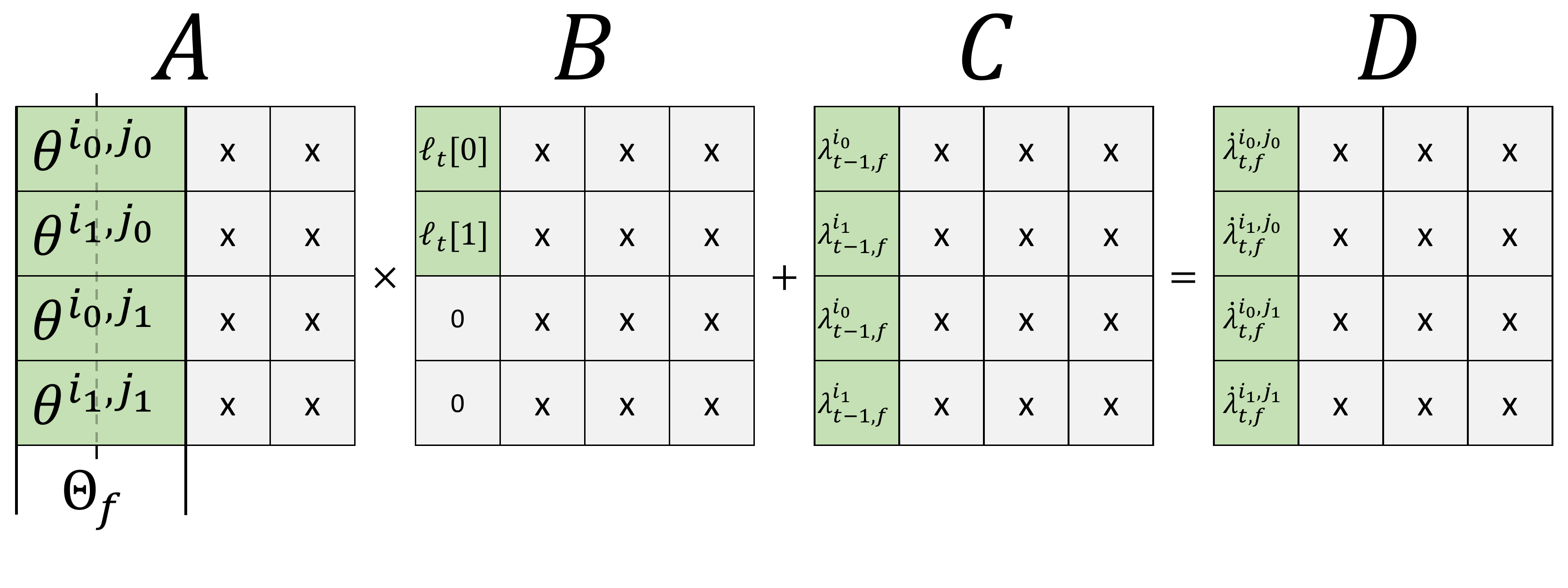} 
	\vskip -6mm
	\caption{Viterbi implementation on $4\times 4$ tensor cores (only one butterfly)}
	\label{fig:tensor4x4:raw}
	\centering
	\includegraphics[width=0.8\textwidth]{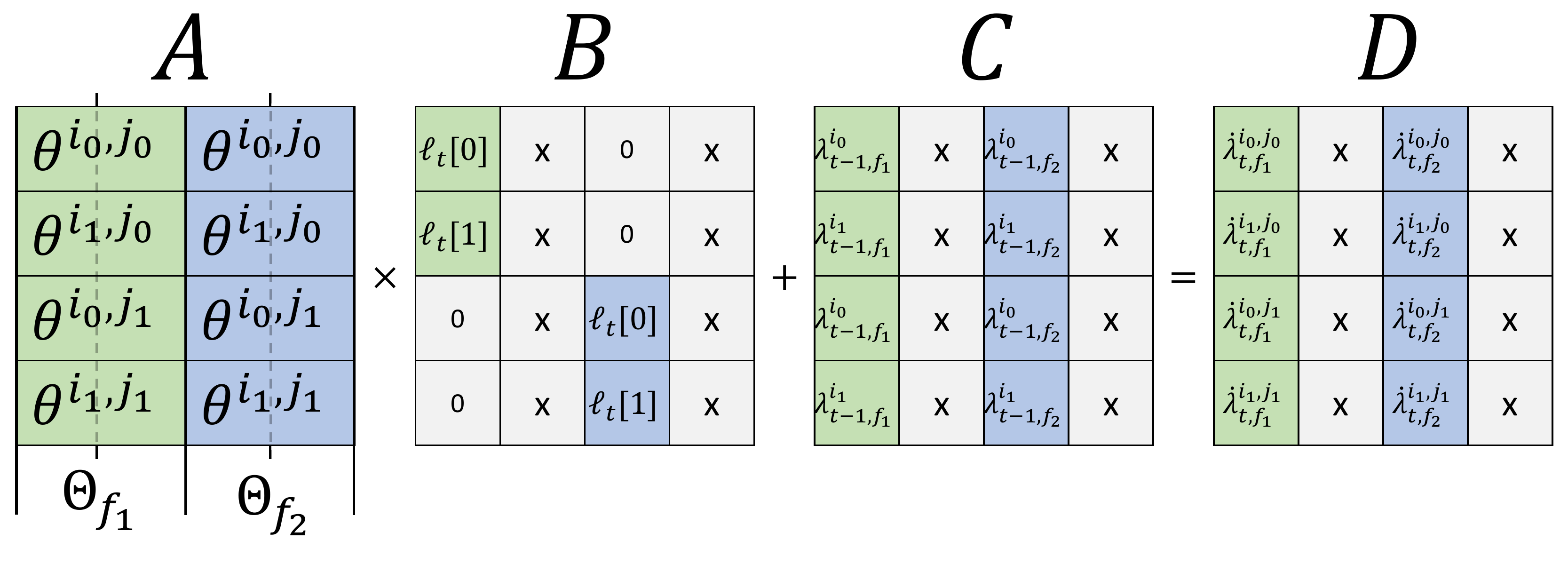} 
	\vskip -3mm
	\caption{Viterbi implementation on $4\times 4$ tensor cores. Two butterflies per 4x4 tensor.}
	\label{fig:tensor4x4:simple}
	\centering
	\includegraphics[width=0.8\textwidth]{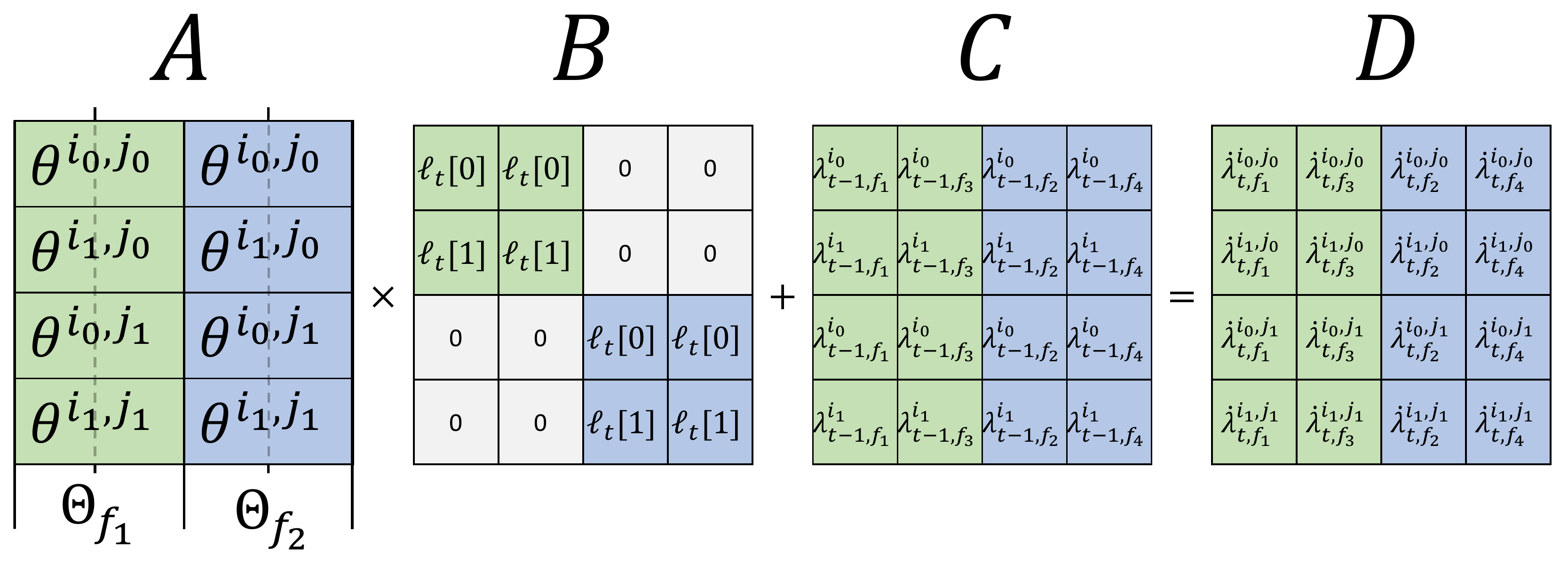} 
	\vskip -3mm	
	\caption{Viterbi implementation on $4\times 4$ tensor cores while using an optimization that employs more entries of the matrices, and also, four butterflies per 4x4 tensor.}
	\label{fig:tensor4x4:opt}
\end{figure*}

As deep learning is one of the most emergent topics in the research world, processors are shifting to be more compatible with such algorithms. Consequently, NVIDIA, one of the most prestigious companies producing parallel processors, has provided Tensor Cores in its Volta series of GPGPU chips that perform matrix multiplication and accumulation incredibly fast. In addition, it is clear that this technology will be improved over these years since deep learning is hungry for processing power. Therefore, making algorithms compatible with tensor cores is not only a high-throughput implementation but also a futuristic approach.

\subsection{Tensor Representation}
In order to implement the Viterbi algorithm on tensor cores, the algorithm should be represented in matrix operations. The first steps of the algorithm, branch metric calculation, is an inner product according to Eq.~\ref{eq:BM}. The operands of the product are a tow vectors of size $\B$. One of them includes LLR values received from channel and the other is contains a sequence of $-1$s and $1$s based on the output of the relevant branch. Since matrix multiplication is some inner products between rows of the first operand and columns of the second operand, four branch metrics of each butterfly can be represented as a multiplication of two matrices, as shown in Eq.~\ref{eq:radix2:branchOutMult} where Eq.~\ref{eq:radix2:branchOutMat} shows the first matrix $\BOBMAT{\BUTI}$ and Eq.~\ref{eq:radix2:llrMat} shows the second one $\LLR$. The result is $\BMMAT_{t,\BUTI}$ that is a matrix containing $\BM$ values. 

\begin{equation}
\label{eq:radix2:branchOutMult}
\BMMAT_{t,\BUTI} = \BOBMAT{\BUTI} \times \LLR_t
\end{equation}

Eq.~\ref{eq:radix2:branchOutMat} shows $\BOBMAT{\BUTI}$ that has four rows regarding output bits of four branches of the butterfly $\BUTI$.

\begin{eqnarray}
\label{eq:radix2:branchOutMat}
\BOBMAT{\BUTI} =
\begin{bmatrix}
\BOBL{i_0,j_0}[0] & \cdots & \BOBL{i_0,j_0}[\B-1] \\
\BOBL{i_1,j_0}[0] & \cdots & \BOBL{i_1,j_0}[\B-1] \\
\BOBL{i_0,j_1}[0] & \cdots & \BOBL{i_0,j_1}[\B-1] \\
\BOBL{i_1,j_1}[0] & \cdots & \BOBL{i_1,j_1}[\B-1] \\
\end{bmatrix}\\ 
\BOBL{i,j}[b] = (-1)^{\BOB{\BUTI}{i,j}[b]}
\end{eqnarray}

The matrix $\LLR_t$ is shown in Eq.~\ref{eq:radix2:llrMat}. It contains only one column because the LLR values are shared for all branches.

\begin{equation}
\LLR_t =
\begin{bmatrix}
\llr_t[0] \\
\vdots \\
\llr_t[\B-1] \\
\end{bmatrix}
\label{eq:radix2:llrMat}
\end{equation}

The result of the production $\BMMAT_{t,\BUTI}$ is a vector of size $4$ that consists of four branch metrics of butterfly $\BUTI$. Proceeding the algorithm, each of two left states should be added to two branch metrics and each result will be a potential path metric for one of right states. Thus, each right state will have two potential path metrics that each of them has come from one of left states. As Eq.~\ref{eq:radix2} shows this procedure as matrix operations, $\BMMAT_{t,\BUTI}$ will be added to another matrix containing the path metrics of the left states. The matrix $\PMMAT_{t-1,\BUTI}$ in Eq.~\ref{eq:radix2:PMMat} is declared for this purpose. $\PMMAT_{t-1,\BUTI}$ should be filled according to $\BOBMAT{\BUTI}$. Each row of the matrix $\PMMAT_{t-1,\BUTI}$ should contain the path metric of the left state of the branch associated with the same row of $\BOBMAT{\BUTI}$.  The result of this addition $\DPMMAT_{t,\BUTI}$ is a vector of size $4$ containing potential path metrics for two right states of the butterfly. As each right state has two incoming branches, it has two potential path metrics. Thus, $\DPMMAT_{t,\BUTI}$ has four rows as shown in  Eq.~\ref{eq:radix2:PMMat}. Each row is a potential path metric of a right state, resulted by the branch associated with the same row of $\BOBMAT{\BUTI}$. 

\begin{equation}
\DPMMAT_{t,\BUTI} = \BMMAT_{t,\BUTI} + \PMMAT_{t-1,\BUTI}
\label{eq:radix2}
\end{equation}

\begin{equation}
\PMMAT_{t-1,\BUTI} =
\begin{bmatrix}
\vphantom{\begin{matrix} a \\ a \end{matrix}} \PM_{t-1,f}^{i_0} \\
\vphantom{\begin{matrix} a \\ a \end{matrix}} \PM_{t-1,f}^{i_1} \\
\vphantom{\begin{matrix} a \\ a \end{matrix}} \PM_{t-1,f}^{i_0} \\
\vphantom{\begin{matrix} a \\ a \end{matrix}} \PM_{t-1,f}^{i_1} \\
\end{bmatrix}
\hphantom{aaaa}
\DPMMAT_{t,\BUTI} =
\begin{bmatrix}
\vphantom{\begin{matrix} a \\ a \end{matrix}} \DPM_{t,f}^{i_0,j_0} \\
\vphantom{\begin{matrix} a \\ a \end{matrix}} \DPM_{t,f}^{i_1,j_0} \\
\vphantom{\begin{matrix} a \\ a \end{matrix}} \DPM_{t,f}^{i_0,j_1} \\
\vphantom{\begin{matrix} a \\ a \end{matrix}} \DPM_{t,f}^{i_1,j_1} \\
\end{bmatrix}
\label{eq:radix2:PMMat}
\end{equation}

\begin{equation}
\PMMAT_{t,\BUTI} =
\begin{bmatrix}
\vphantom{\begin{matrix} a \\ a \end{matrix}} \PM_{t,f}^{j_0} \\
\vphantom{\begin{matrix} a \\ a \end{matrix}} \PM_{t,f}^{j_1} \\
\vphantom{\begin{matrix} a \\ a \end{matrix}} \PM_{t,f}^{j_0} \\
\vphantom{\begin{matrix} a \\ a \end{matrix}} \PM_{t,f}^{j_1} \\
\end{bmatrix}
=
\begin{bmatrix}
\vphantom{\begin{matrix} a \\ a \end{matrix}} max(\DPM_{t,f}^{i_0,j_0}, \DPM_{t,f}^{i_1,j_0}) \\
\vphantom{\begin{matrix} a \\ a \end{matrix}} max(\DPM_{t,f}^{i_0,j_1}, \DPM_{t,f}^{i_1,j_1}) \\
\vphantom{\begin{matrix} a \\ a \end{matrix}} max(\DPM_{t,f}^{i_0,j_0}, \DPM_{t,f}^{i_1,j_0}) \\
\vphantom{\begin{matrix} a \\ a \end{matrix}} max(\DPM_{t,f}^{i_0,j_1}, \DPM_{t,f}^{i_1,j_1}) \\
\end{bmatrix}
\label{eq:radix2:PMMatFinal}
\end{equation}

The multiplication and addition mentioned above, are equivalent to the forward procedure of the Viterbi algorithm. Eq.~\ref{eq:radix2} shows whole the computation. After this calculation, the winner path of each state should be selected and stored according to potential final path metrics. This calculation that forms $\PMMAT_{t,\BUTI}$, path metrics of destination states, is shown in Eq.~\ref{eq:radix2:PMMatFinal}. The last step of the algorithm, traceback, should be done afterward. Traceback should be done in its ordinary manner and can not be represented using matrix multiplication.

%

\subsection{Mapping to Tensor Cores}
The similarity between Eq.~\ref{eq:radix2} and tensor operation makes the Viterbi algorithm implementable on this novel cores. For the purpose of simplicity, a $4\times4$ tensor core is used at first. The matrices $A$, $B$, $C$ and $D$ in a tensor operation are used to contain matrices $\BOBMAT{\BUTI}$, $\LLR_t$, $\PMMAT_{t-1,\BUTI}$ and $\DPMMAT_{t,\BUTI}$, respectively. As an example, Fig.~\ref{fig:tensor4x4:raw} shows how to fill these matrices for $\B=2$. 

In Fig.~\ref{fig:tensor4x4:simple}, two butterflies are processed in a single operation as discriminated with different colors. It can also be noticed some entries do not play any role and they should be assigned a job to reduce the number of unused entries as much as possible.

Eq.~\ref{eq:radix2:boRelatBegin} to Eq.~\ref{eq:radix2:boRelatEnd} imply that the matrix $\BOBMAT{\BUTI}$ can be defined with its first row and the other rows can be obtained using that. Hence, there are $2^{\B}$ distinct valid matrices as $\BOBMAT{\BUTI}$ because it has $\B$ columns. On the other hand, $2^{\K-1} \div 2 = 2^{\K-2}$ butterflies exist in each stage of trellis. It means that some butterflies can have the same $\BOBMAT{\BUTI}$ if $2^{\B} < 2^{\K-2}$. For example, in Viterbi decoder with $k=7$ and $\B=2$ generative polynomials $(171,133)$, $2^{\K-2}=32$ butterflies are uniformly distributed between $2^{\B} = 4$ different valid matrices for $\BOBMAT{\BUTI}$, $32 \div 4 = 8$ butterflies for each matrix. Consequently, empty columns of matrices $B$, $C$ and $D$ can be filled with values associated with another butterfly of the same $\BOBMAT{\BUTI}$ as matrix $A$ as depicted in Fig.~\ref{fig:tensor4x4:opt}. As a result, $4$ butterflies can be processed in a single operation.

NVIDIA has provided developers with its API that employs tensor cores inside kernels. This API, however, enables developers to use not $4\times 4$ tensor cores but $16\times16$ ones. Thus, in order to utilize tensor cores, $4\times4$ approach should be implemented 4 times across the main diagonal of the matrices. In this way, $16$ butterflies are processed in a single $16\times16$ tensor operation. In other words, each stage of the decoder can be performed by invoking tensor operation $2^{\K-2} \div 16 = 2^{\K-6}$ times, i.e., $Q=2^{\K-6}$. For $\K=7$, $Q=2$.

%
%
%
%
%
%
%
%
%
%
%
%
%

%% file: radix4.tex
\section{Generalization to Radix-$2^{\RADIXI}$}
\label{sec:radixn}

Radix-$2^{\RADIXI}$ in Viterbi algorithm is a concept like that of FFT algorithm, which is a pattern larger than regular butterflies. A butterfly is a pattern in a trellis, that is one-stage-wide, ranging from stage $t$ to $t+1$. It has $2$ states in each stage, per butterfly. Patterns similar to butterflies but larger can also be found in a trellis. These patterns that are extended versions of butterflies are called dragonflies. Patterns that are A $\RADIXI$-stage-wide are called Radix-$2^{\RADIXI}$ dragonflies. In other words, $\RADIXI$ is width of the dragonfly (the number of stages that it processes and moves forward). $2^{\RADIXI}$ is the number of states in every dragonfly in every stage. Since $2^{\K-1}$ is the number of all states per stage, there are $2^{\K-1} \div 2^{\RADIXI} = 2^{\K-\RADIXI-1}$ dragonflies per stage. Fig.~\ref{fig:radixGeneral} shows three different dragonflies. The green one is a regular butterfly that, in this literature, can be called a Radix-2 dragonfly. This green dragonfly combined with other three blue ones forms a Radix-4 dragonfly that ranges from stage $0$ to $2$ and has $4$ states in every stage. Including brown branches, a Radix-8 dragonfly is depicted.

\begin{figure}[h]
	\centering
	\includegraphics[width=0.9\columnwidth]{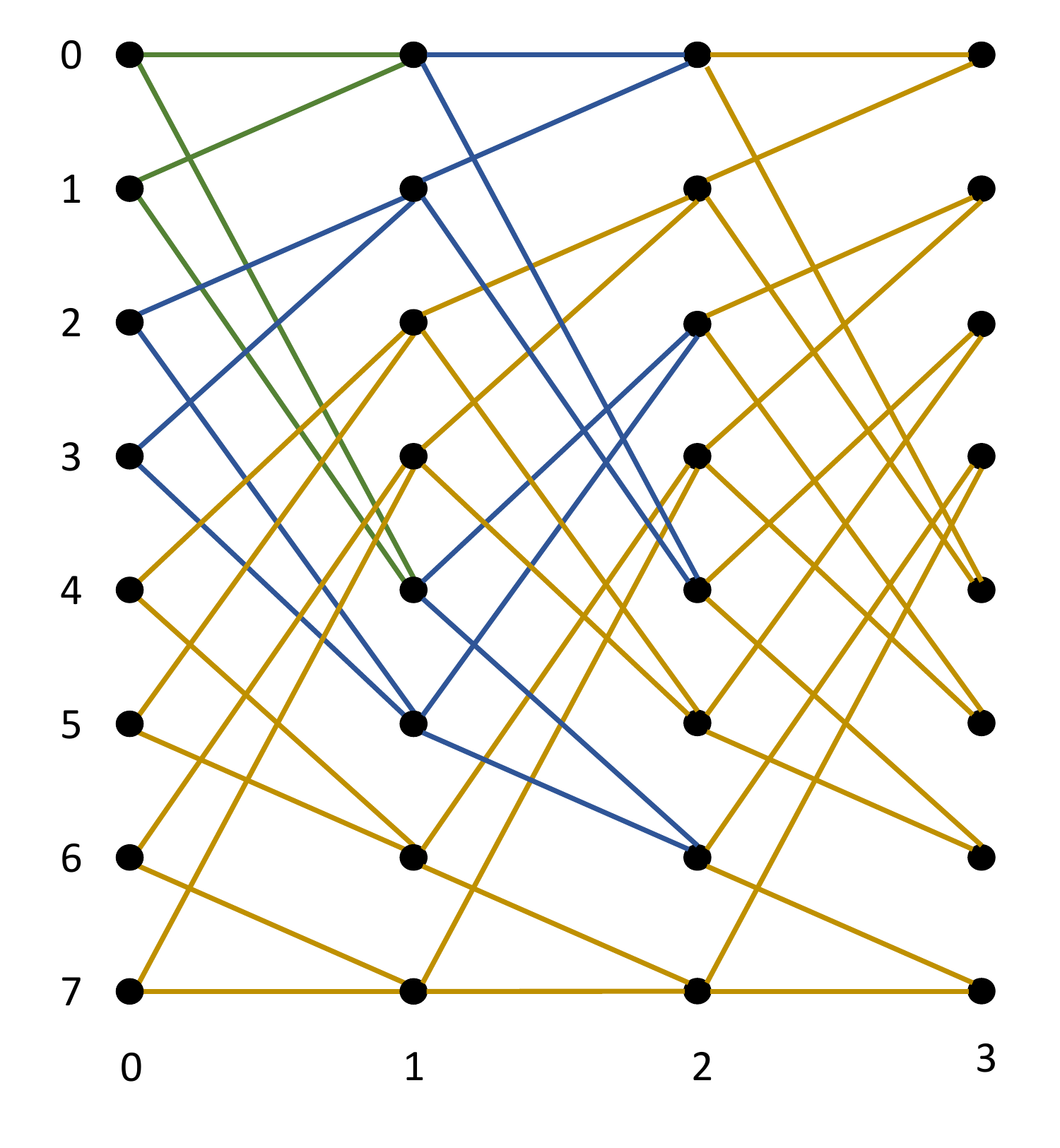} 
	\caption{dragonfly pattern with three different values of $\RADIXI$}
	\label{fig:radixGeneral}
\end{figure}

Assuming that a part of a trellis from stage $t$ to stage $t+\RADIXI$ is extracted, it is a disconnected graph and dragonflies are isolated sub-graphs constituting it. It means that in a Radix-$2^{\RADIXI}$ dragonfly, which has $2^{\RADIXI}$ states in each of $(\RADIXI+1)$ stages from $t$ to $t+1$, there is no path between a pair of states consisting a branch or a state outside the dragonfly. In another perspective, left states of a Radix-$2^{\RADIXI}$ dragonfly are a subset of all states in that stage that will reach a subset of the same size until $\RADIXI$ next stages.

TODO: Higher radix means larger number of stages which eventually leads to more efficiency as we will see in the next section. To this end we need to find a function that help us find global state index in a given position of a dragonfly with a provided index. It means that global state index is a function of local state index and local stage index that represent the local position, and dragonfly index. This relation is the general form of Eq.~\ref{eq:radix2:S}.

Global state index = func ( Dragonfly index, Local state index, Dragonfly stage )

\begin{theorem}
	Let $0 \le b < 2^{\K-\RADIXI-1}$. Every subset of states in stage $t$ defined as $S = \{b+x \mid 0 \le x < 2^{\RADIXI}\}$ contains all left states of Radix-$2^{\RADIXI}$ dragonfly.
\end{theorem}

\begin{proof}
	$S$ consists of all states that are different only in their $\RADIXI$ LSB bits and the rest of their binary representations are the same. Thus, $S$ has $\RADIXI$ free bits and it has $2^{\RADIXI}$ members. Considering all branches from these states, they will reach to $S'$, a subset of states in stage $t+1$. Regarding the encoder shift-register in that stage, one MSB bit is shifted in which is a free bit and members of $S'$ can hold both $0$ and $1$ in that position. Meanwhile, one LSB bit is shifted out and one of $\RADIXI$ initial free bits are removed. Therefore, states of $S'$ has still $\RADIXI$ free bits. In other words, the size of $S'$ is $2^{\RADIXI}$ which is the same as $S$. Similarly, in next stages, one free bit is added and one is removed. Until stage $t+\RADIXI$ that all $\RADIXI$ free bits are shifted out, all branches reach to a subset of the same size and it is what a Radix-$2^{\RADIXI}$ dragonfly looks like.
\end{proof}

\begin{corollary}[Model of Bubble and Fluid]
	Fig.~\ref{fig:bubble} shows paths of a Radix-$2^{\RADIXI}$ dragonfly in the perspective of encoder shift-register. In all dragonflies, a fixed part is moving from the first position in stage $t$ to the last position in stage $t+\RADIXI$ like a bubble moves in fluid. The fluid is divided into two parts, pre-bubble and post-bubble, and in spite of the bubble, can hold all possible values per stage.
\end{corollary}

In order to obtain global state indexes, two operators should be defined. The first operator shown in Eq.~\ref{eq:bit} extracts a portion of binary representation of the number $x$. In Eq.~\ref{eq:bit}, the portion is defined with its first bit index $b$ and the first index after it $a$. Bit indexes increase from LSB to MSB and start with one.

\begin{equation}
x_{b\BITIDX a} = extracted~portion~of~x~from~~bit~a~to~bit~b\\
\label{eq:bit}
\end{equation}


For example, assuming $39$ as $x$ with binary representation $100111$, $x_{4\BITIDX 1}$ will be the portion of the binary representation from $4$ to $2$ that is $011$ and equals to $3$. Similarly, $x_{4\BITIDX 0}$ will be $0111$ or $7$.


Eq.~\ref{eq:shift} shows the other operation that is putting $y$ zeros after the number $x$ in its binary representation. It is similar to left shift instruction.

\begin{equation}
x \ll y = x \times 2^y
\label{eq:shift}
\end{equation}

In the following, a general relation to find global state index will be provided and will be proved.

\begin{theorem}
	Given local stage index $x$, local state index $y$ and dragonfly index $\BUTI$, global state index $s$ is obtained using Eq.~\ref{eq:general:S}.
	\begin{align}
	\forall 0 \le x \le \RADIXI , 0 \le y < 2^{\RADIXI} : \s_{y,x} &= \underbrace{[y_{\RADIXI \BITIDX \RADIXI-x} \ll (\K-x-1)]}_{pre-bubble} \\
	&+ \underbrace{[\BUTI \ll (\RADIXI-x)]}_{bubble} \\ 
	&+ \underbrace{[y_{\RADIXI-x-1 \BITIDX 0}]}_{post-bubble}
	\label{eq:general:S}
	\end{align}
\end{theorem}


\begin{proof}
	The bubble value is fixed in all states and stages of a Radix-$2^{\RADIXI}$ dragonfly and can be considered as an identifier for a dragonfly. The more this value is the greater all states are and the lower the dragonfly is placed in the trellis. The dragonfly index, $\BUTI$, similarly increases from top to bottom. That is why it can be claimed that the bubble value equals to $\BUTI$.
	
	The fluid that is $\RADIXI$-bit-wide can take all $2^{\RADIXI}$ possible values in each stage of a specific Radix-$2^{\RADIXI}$ dragonfly. Furthermore, there are $2^{\RADIXI}$ states in each stage and the more the value of the fluid is, the lower the state is. As a result, the value of the fluid is the local state of the dragonfly denoted as $y$ and ranging from $0$ at the top to $2^{\RADIXI}-1$ at the bottom.
	
	Fig.~\ref{fig:bubble} shows that in stage $t+x$, $x$ bits of the fluid is dedicated to pre-bubble and the rest is post-bubble. Using the operator defined in Eq.~\ref{eq:bit}, pre-bubble and post-bubble are equal to $y_{\RADIXI \BITIDX \RADIXI-x}$ and $y_{\RADIXI-x-1 \BITIDX 0}$, respectively.
	
	Global state index is a concatenation of pre-bubble, bubble and post-bubble. All three parts are put on their position using the shift operator defined in Eq.~\ref{eq:shift}.
\end{proof}

\begin{theorem}
	The connections of a Radix-$2^{\RADIXI}$ dragonfly is like a $2^{\RADIXI}$-state trellis with $\K=\RADIXI+1$.
\end{theorem}

\begin{proof}
	connections of states are defined by fluid since it delegates to local states. Therefore, the bubble does not have any effect and can be removed. As a result the length of the shift-register will be $\RADIXI$ which is equivalent to $\K=\RADIXI+1$.
\end{proof}

\section{Radix-4 patterns in the trellis}
\label{sec:radix4}

Implemented on $16 \times 16$ tensor cores, the prior method fills less than half of the matrix entries. Hence, to use most of the matrix entries and consequently most of the tensor core power, a reformed observation of trellis is approached that employs a wider pattern in the branches. This modification employs Radix-4 dragonflies that are related to a specific mode of Radix-$2^{\RADIXI}$ with $\RADIXI=2$. As a result, not only it utilizes $16\times16$ tensor cores to a great extent but it also benefits from other improvements in memory access that will be discussed.

Fig.~\ref{fig:radix4}(a) shows an 8-state trellis with two Radix-4 dragonflies. It can be observed that each Radix-4 dragonfly consists of $4$ butterflies. There are $2^{k-1} \div 4 = 2^{k-3}$ dragonflies per stage since each of them contains $4$ states per stage.


Fig.~\ref{fig:radix4}(b) shows a sole dragonfly that has connections like a 4-state trellis with states from $0$ to $3$. Global state indexes of left states and right states are denoted by $i$ and $j$ like butterflies and $m$ is used for middle states.

In order to obtain global state indexes from local state indexes in a Radix-4 dragonfly, Eq.~\ref{eq:general:S} should be used. Eq.~\ref{eq:radix4:S} shows the results for all $12$ states where $\BUTI$ is the dragonfly index and symbols used for states are related to Fig.\ref{fig:radix4}(b).

\begin{figure}[h]
	\centering
	\includegraphics[width=0.9\columnwidth]{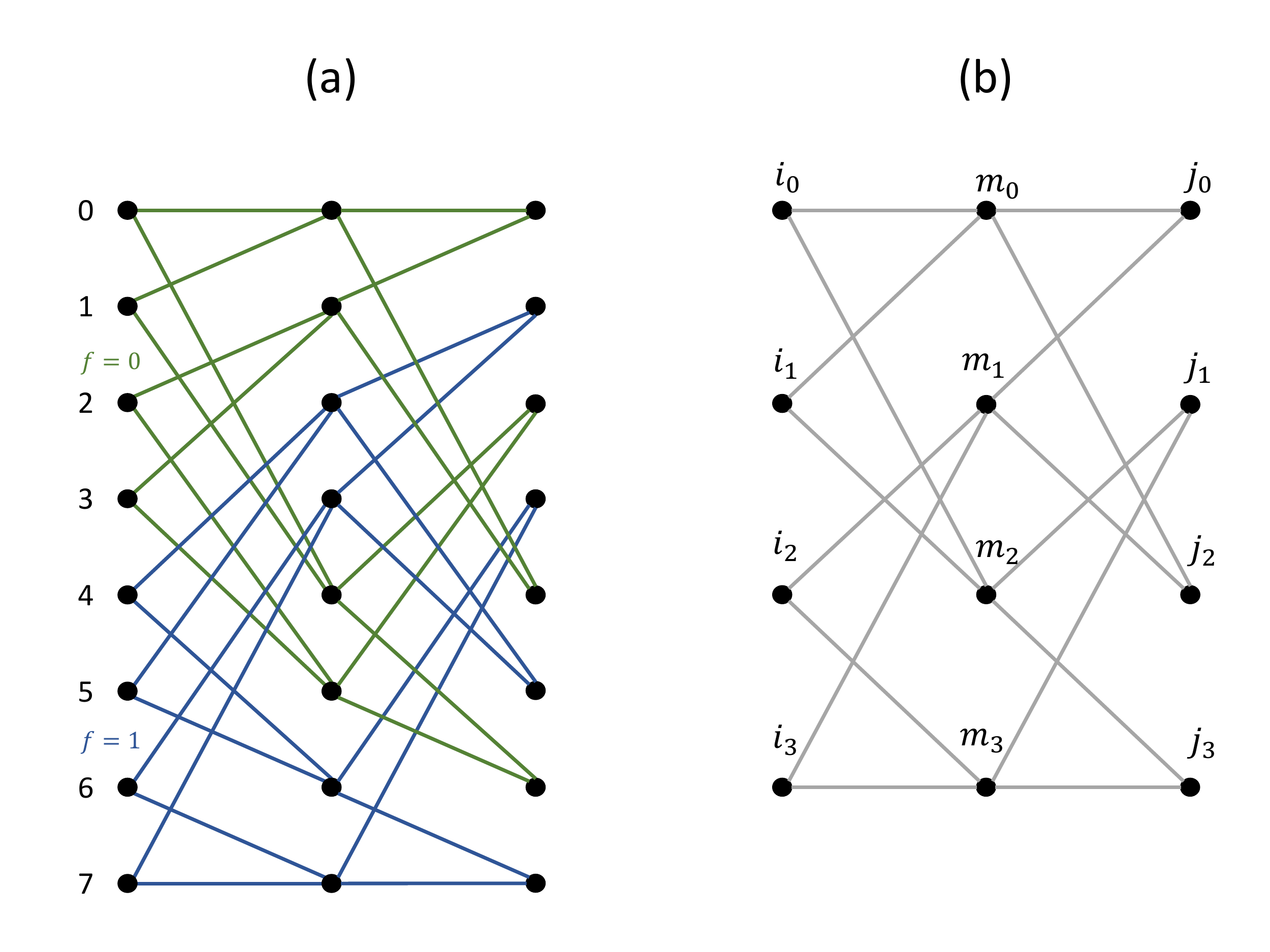} 
	\caption{(a) Radix-4 dragonflies in a trellis with 8 states. (b) a single dragonfly showing the connections and the parameters with dragonfly-wise indexing}
	\label{fig:radix4}
\end{figure}

\begin{equation}
\begin{matrix*}[l]
i_0=4\BUTI		& , & m_0=2\BUTI 			& , & j_0=\BUTI\\
i_1=4\BUTI+1 	& , & m_1=2\BUTI+1		 	& , & j_1=\BUTI+2^{\K-3}\\
i_2=4\BUTI+2 	& , & m_2=2\BUTI+2^{\K-2} 	& , & j_2=\BUTI+2\times2^{\K-3}\\
i_3=4\BUTI+3 	& , & m_3=2\BUTI+2^{\K-2}+1	& , & j_3=\BUTI+3\times2^{\K-3}\\
\end{matrix*}
\label{eq:radix4:S}
\end{equation}

\begin{theorem}
	In a Radix-4 dragonfly, there is only one path between every left state and every right state.
\end{theorem}

\begin{proof}
	 Fig.~\ref{fig:tree}(a) shows a sub-graph of a Radix-4 dragonfly containing all paths originating from a specific left state. On the other hand, Fig.~\ref{fig:tree}(b) shows a similar sub-graph including all paths towards a particular right state. It can be observed that all of these sub-graphs are trees. It means only one path can start from a specific left state and end at a specific right one.
\end{proof}

\begin{figure}[h]
	\centering
	\includegraphics[width=0.9\columnwidth]{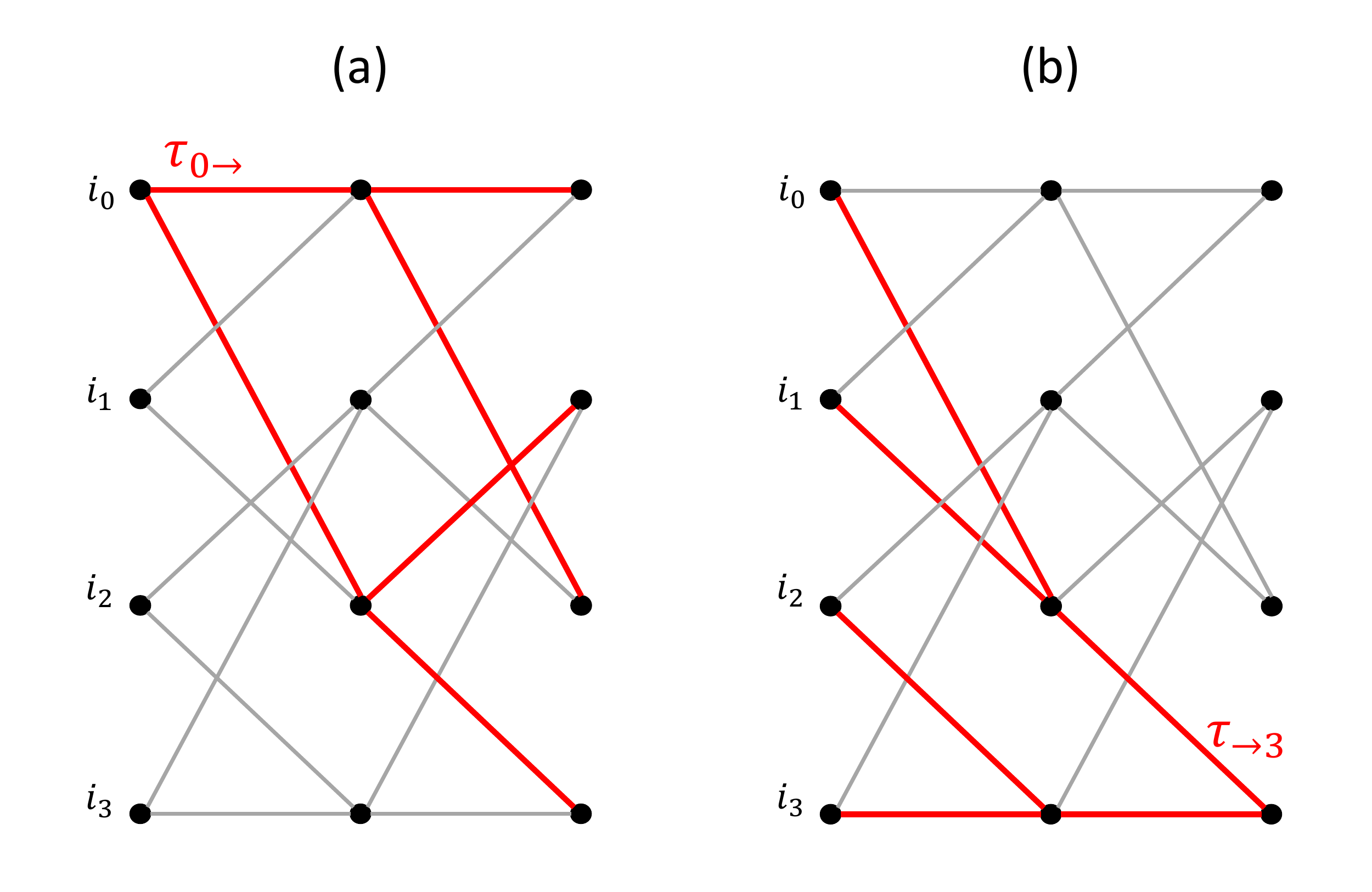} 
	\caption{tree patterns of a 4-state dragonfly. (a) a initial-stage-rooted tree. (b) a final-stage-rooted tree}
	\label{fig:tree}
\end{figure}

Such trees are denoted as $\TREEI{i}$ for trees rooted at the left stage and $\TREEF{i}$ for trees rooted at the right stage where $i$ is the root state. For instance, Fig.~\ref{fig:tree} shows $\TREEI{0}$ and $\TREEF{3}$ in a 4-state trellis.

\begin{corollary}
	 In a Radix-4 dragonfly, the unique path between each pair of left and right states can be considered as a super-branch and middle states will be eliminated. Fig.~\ref{fig:bipartite}, which depicts such representation, shows that the dragonfly is a complete bipartite graph.
\end{corollary}

As Fig.~\ref{fig:bipartite} shows, super-branch outputs are denoted as $\BOBRF{\BUTI}{i,j}$ where $i$ and $j$ are respectively initial and final states.

\begin{figure}[h]
	\centering
	\includegraphics[width=0.8\columnwidth]{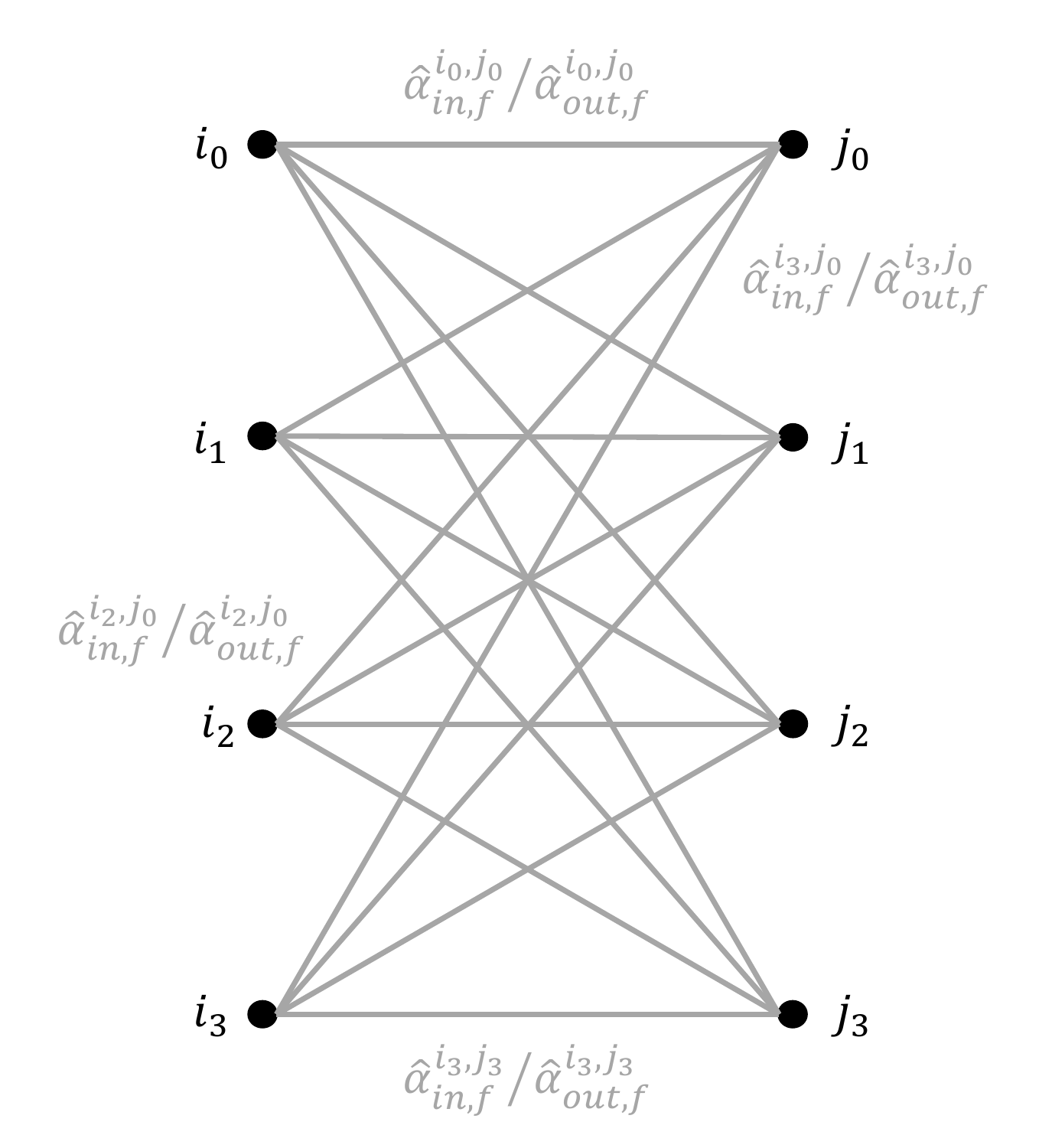} 
	\caption{bipartite representation of a 4-state dragonfly}
	\label{fig:bipartite}
\end{figure}

\begin{theorem}
	All super-branch outputs of a given dragonfly are related to one another. Specifically, the relation between the first super-branch output and other ones is shown in Eq.~\ref{eq:outputRelation}.
	\begin{equation}
	\forall i,j<2^{\RADIXI} \phantom{a} \exists func(x) : \BOBRF{\BUTI}{i,j} = func(\BOBRF{\BUTI}{i_0,j_0})
	\label{eq:outputRelation}
	\end{equation}
\end{theorem}

\begin{proof}
	Branch outputs are calculated with Eq.~\ref{eq:encoder}. In order to explore their relation in a dragonfly branch outputs, a reformed representation should be used. Eq.~\ref{eq:general:output1} shows branch outputs for one of the polynomials and in a specific stage based on the concept of the bubble and fluid. To be expressed in detail, every generative polynomial is divided into four parts, $G_0$, $G_1$, $G_2$ and $G_3$ that are applied to input bit, pre-bubble, bubble and post-bubble, respectively. The length of each part is different in different stages of a dragonfly.
	
	\begin{equation}
	\BOB{\BUTI}{i,j} = G_0.in_t \oplus G_1.B_{pre} \oplus G_2.B \oplus G_3.B_{post}
	\label{eq:general:output1}
	\end{equation}
	
	Since fluid represents local state index, for all branches from zero to zero, both pre-bubble and post-bubble are zero. The input bit is also zero. So
	
	\begin{equation}
	\BOB{\BUTI}{i_0,j_0} = G_2.B
	\label{eq:originalOutput}
	\end{equation}
	
	Using Eq.~\ref{eq:general:output1} and Eq.~\ref{eq:originalOutput}, it can be resulted that
	
	\begin{equation}
	\BOB{\BUTI}{i,j} = G_0.in_t \oplus G_1.B_{pre} \oplus \BOB{\BUTI}{i_0,j_0} \oplus G_3.B_{post}
	\label{eq:general:output2}
	\end{equation}
	
	In other words, in every stage of the dragonfly, all branch outputs can be defined given $\BOB{\BUTI}{i_0,j_0}$. This branch from zero to zero can be called the main branch of that stage. 
	
	Eq.~\ref{eq:general:output2} can be generalized to super-branches. As super-branches of a Radix-$2^{\RADIXI}$ dragonfly is composed of $\RADIXI$ branches, a super branch crossing all zeros can be defined as the main super-branch and the output of all super-branches can be obtained from the output of the main super-branch.	
\end{proof}

%% file: tensor4.tex
\section{Employing Tensor Cores For Radix-4}
\label{sec:tensorRadix4}

\subsection{Reformed Formulation}
Preceding the implementation of the Viterbi algorithm with Radix-4 approach, Eq.~\ref{eq:BM} to Eq.~\ref{eq:SP0} should be revised. Eq.\ref{eq:radix4:BM} to Eq.~\ref{eq:radix4:SP0} show the modification. Eq.~\ref{eq:radix4:BM} shows the metric of the super-branch that is twice a branch. Therefore, the summation length is considered $2\B$. Eq.~\ref{eq:radix4:PM} shows the forward procedure to update path metric values using super-branch metrics. The important fact shown in Eq.~\ref{eq:radix4:PM} is that the forward procedure, which is iterative, is performed two stages per iteration. In other words, the number of iterations is half of the Radix-2 version of the algorithm. Eq.~\ref{eq:radix4:SP0} shows storing survivor paths that will be used in the traceback step. As performed after path metric calculation, it is also done two stages per iteration. Thus, the traceback step will be also done two stages per iteration that leads to less access to memory.

\begin{equation}
\BMRF^{i,j}_t = \sum_{b=0}^{2\B-1} (-1)^{\displaystyle \BORF^{i,j}[b]} \times \llr_t[b]
\label{eq:radix4:BM}
\end{equation}

\begin{equation}
\PM^j_t = \max_{i \in \text{prv}(\text{prv}(j))} \big( 
\PM^{i}_{t-2} + \BMRF^{i,j}_t  
\big)
\label{eq:radix4:PM}
\end{equation}

\begin{equation}
\SP^j_t = \argmax_{i \in \text{prv}(\text{prv}(j))} \big( 
\PM^{i}_{t-2} + \BMRF^{i,j}_t  
\big)
\label{eq:radix4:SP0}
\end{equation}

\subsection{Tensor Representation}
The matrix representation of the algorithm is similar to Radix-2 except for the size of the matrices. A Radix-4 dragonfly, in its bipartite representation, consists of 16 super-branches with $2\B$-bit-wide outputs. As a result, $\BOBMAT{\BUTI}$ is a $16\times 2\B$ matrix and both $\PMMAT_{t-1,\BUTI}$ and $\DPMMAT_{t,\BUTI}$ are $16\times 1$ matrices. $\LLR_t$ is the same as the Radix-2 version while it has $2\B$ rows. Eq.~\ref{eq:radix4:branchOutMat} shows the general format of $\BOBMAT{\BUTI}$. This matrix can be considered as it is a combination of some partial matrices $P_j$.  $P_j$ is dedicated to $\TREEF{j}$ which is a right-stage-rooted tree. Eq.~\ref{eq:radix2:PMMatI} and Eq.~\ref{eq:radix2:PMMatF} show $\PMMAT_{t-1,\BUTI}$ and $\DPMMAT_{t,\BUTI}$ that are arranged according to $\BOBMAT{\BUTI}$. 


\begin{equation}
\BOBMATRF{\BUTI} =
\begin{bmatrix}
P_0 \\
P_1 \\
\vdots \\
P_{(2^{\RADIXI})-1}
\end{bmatrix}
\left. \Bigg| \right.  P_j =
\underbrace{ \begin{bmatrix}
\BOBLRF{i_0,j}[0] & \BOBLRF{i_0,j}[1] & \cdots \\
\BOBLRF{i_1,j}[0] & \BOBLRF{i_1,j}[1] & \cdots \\
\BOBLRF{i_2,j}[0] & \BOBLRF{i_2,j}[1] & \cdots \\
\vdots & \vdots & \ddots
\end{bmatrix}}_{\RADIXI \times \B~columns}
\left.\vphantom{\begin{matrix} a\\ a\\ a\\ a\\ a \end{matrix}}\right\}2^{\RADIXI}
\label{eq:radix4:branchOutMat}
\end{equation}

\begin{equation}
\PMMAT_{t-2,\BUTI} =
\begin{bmatrix}
P_0 \\
P_1 \\
\vdots \\
P_{2^{\RADIXI}-1}
\end{bmatrix}
\left. \Bigg| \right.  P_j =
\begin{bmatrix}
\vphantom{\begin{matrix} a \\ a \end{matrix}} \PM_{t-2,f}^{i_0} \\
\vphantom{\begin{matrix} a \\ a \end{matrix}} \PM_{t-2,f}^{i_1} \\
\vdots \\
\end{bmatrix}
\left.\vphantom{\begin{matrix} a\\ a\\ a\\ a\\ a\\ a \end{matrix}}\right\}2^{\RADIXI}
\label{eq:radix2:PMMatI}
\end{equation}

\begin{equation}
\DPMMAT_{t,\BUTI} =
\begin{bmatrix}
P_0 \\
P_1 \\
\vdots \\
P_{2^{\RADIXI}-1}
\end{bmatrix}
\left. \Bigg| \right.  P_j =
\begin{bmatrix}
\vphantom{\begin{matrix} a \\ a \end{matrix}} \DPM_{t,f}^{i_0,m,j} \\
\vphantom{\begin{matrix} a \\ a \end{matrix}} \DPM_{t,f}^{i_1,m,j} \\
\vdots
\end{bmatrix}
\left.\vphantom{\begin{matrix} a\\ a\\ a\\ a\\ a\\ a \end{matrix}}\right\}2^{\RADIXI}
\label{eq:radix2:PMMatF}
\end{equation}

\subsection{Mapping to Tensor Cores}
The implementation of the Viterbi algorithm is more compatible with $16\times 16$ tensor cores using Radix-4 approach because one dimension of the matrices is exactly the same as that of tensor cores. Fig.~\ref{fig:tensor16x16:simple} shows a code of rate $\frac{1}{2}$. In such a code, $\BOBMAT{\BUTI}$ is a $16 \times 4$ matrix four of which can be accommodated in matrix $A$. Matrices $B$ and $C$ are filled with $\LLR$ and $\PMMAT_{t-2,\BUTI}$ in four columns. The result matrix $D$ will have four columns of matrices $\DPMMAT_{t,\BUTI}$ in the same columns as the matrix $C$ that is depicted. Fig.~\ref{fig:tensor16x16:simple} illustrates that each matrix operation performs calculations of four dragonflies. Furthermore, there are $2^{\K-1} \div 4 = 2^{\K-3}$ dragonflies in every two stages. As a result, every two stages are processed in $2^{\K-3} \div 4 = 2^{\K-5}$ tensor operations. In other words, $Q=2^{\K-6}$ tensor operations are needed per stage. It is the same as Radix-2. However, less memory access are gained, since the number of iterations are half the Radix-2.

\subsection{Optimization}
\subsubsection{Possibility of Previous Method}
Eq.\ref{eq:outputRelation} shows that only the first row of the matrix $\BOBMAT{\BUTI}$ defines the whole. It means that $2^{\RADIXI \times \B}$ distinct matrices exist in a code, as $\RADIXI \times \B$ is the number of columns in $\BOBMAT{\BUTI}$. On the other hand, since each dragonfly accommodates $2^{\RADIXI}$ states per stage, there are $2^{\K-1} \div 2^{\RADIXI} = 2^{\K-1-\RADIXI}$ dragonflies in each stage. As a result, If there are more than $2^{\RADIXI \times \B}$ dragonflies that means $\RADIXI \times \B < \K-1-\RADIXI$, some dragonflies have the same $\BOBMAT{\BUTI}$. In that case, unused columns in Fig.~\ref{fig:tensor16x16:simple} can be used for a dragonfly with the same $\BOBMAT{\BUTI}$ like the method used in Radix-2 and depicted in Fig.~\ref{fig:tensor4x4:opt}. Implemented in the code of $k=7$ and $polynomial=(171,133)$ which is a most used, all $2^{7-1-2}=16$ dragonflies are all $2^{2 \times 2}=16$ possible ones and such optimization cannot be utilized.

\subsubsection{A Novel Method}
In order to optimize this particular code, a novel strategy is employed that fills whole the matrices $C$ and $D$. Fig.~\ref{fig:BOMATs} is a table of 16 column that each of them is a $\BOBMAT{\BUTI}$ of a dragonfly. each entry is the decimal representation of the 4-bit super-branch output. Columns with the same color are the same set with different ordering. In other words, they are permuted version of each other. Similar columns fall into the same dragonfly group. Dragonfly groups are denoted by $\BG$. Eq.~\ref{eq:bg0} to Eq.~\ref{eq:bg3} show dragonfly groups in Fig.\ref{fig:BOMATs}.

\begin{align}
\label{eq:bg0}	\BG_0 &= \{\BOBMAT{0}, \BOBMAT{2}, \BOBMAT{8}, \BOBMAT{10}\}\\
\label{eq:bg1}	\BG_1 &= \{\BOBMAT{1}, \BOBMAT{3}, \BOBMAT{9}, \BOBMAT{11}\}\\
\label{eq:bg2}	\BG_2 &= \{\BOBMAT{4}, \BOBMAT{6}, \BOBMAT{12}, \BOBMAT{14}\}\\
\label{eq:bg3}	\BG_3 &= \{\BOBMAT{5}, \BOBMAT{7}, \BOBMAT{13}, \BOBMAT{15}\}
\end{align}

In a dragonfly group, tensor operation can be performed by only one $\BOBMAT{\BUTI}$ if $\PMMAT_{t-2,f}$ and $\DPMMAT_{t,f}$ of other dragonflies are permuted. To be explained in detail for $\BG_0$ as an example, $\BOBMAT{0}$ can be used for all four dragonflies in the group. In order to do that, the permutation function mapping $\BOBMAT{f}$ to $\BOBMAT{0}$ should be applied to $\PMMAT_{t-2,f}$ and $\DPMMAT_{t,f}$. 
Consequently, all 16 dragonflies can be processed using four matrices as $\BOBMAT{}$ and it means that the calculation of all 16 dragonflies, the entire trellis, can be performed in a single tensor operation as depicted in Fig.~\ref{fig:tensor16x16:opt}. Since this operation is for two stages, $Q=0.5$ tensor operation is needed per stage.

\subsubsection{A Deep Interpretation}
Although this optimization seems complicated, an intuitive description of this permutation is achieved after a profound observation. In every two peer matrices of the same dragonfly group, the first four rows are the same subsets and it is also true for following four rows. It can also be observed that the permutation for all subsets are the same. In other words, in all pairs of $\BOBMAT{f}$s in a dragonfly group, $P_j$s that are subsets of $\BOBMAT{f}$s according to Eq.~\ref{eq:radix4:branchOutMat}, are the same subsets and the permutation function is also the same for different $j$s. Since $P_j$ is associated with $\TREEF{j}$ and the permutation in $P_j$ is equivalent to permutation in initial states of $\TREEF{j}$ and also this permutation is the same for all trees, it is enough to consider that initial states of dragonflies are permuted and then all four dragonflies of a dragonfly group will have the same $\BOBMAT{}$. The permutation in initial states for all dragonflies is shown in Fig.~\ref{fig:permutation}.

\begin{figure}[t]
	\centering
	\includegraphics[width=0.9\columnwidth]{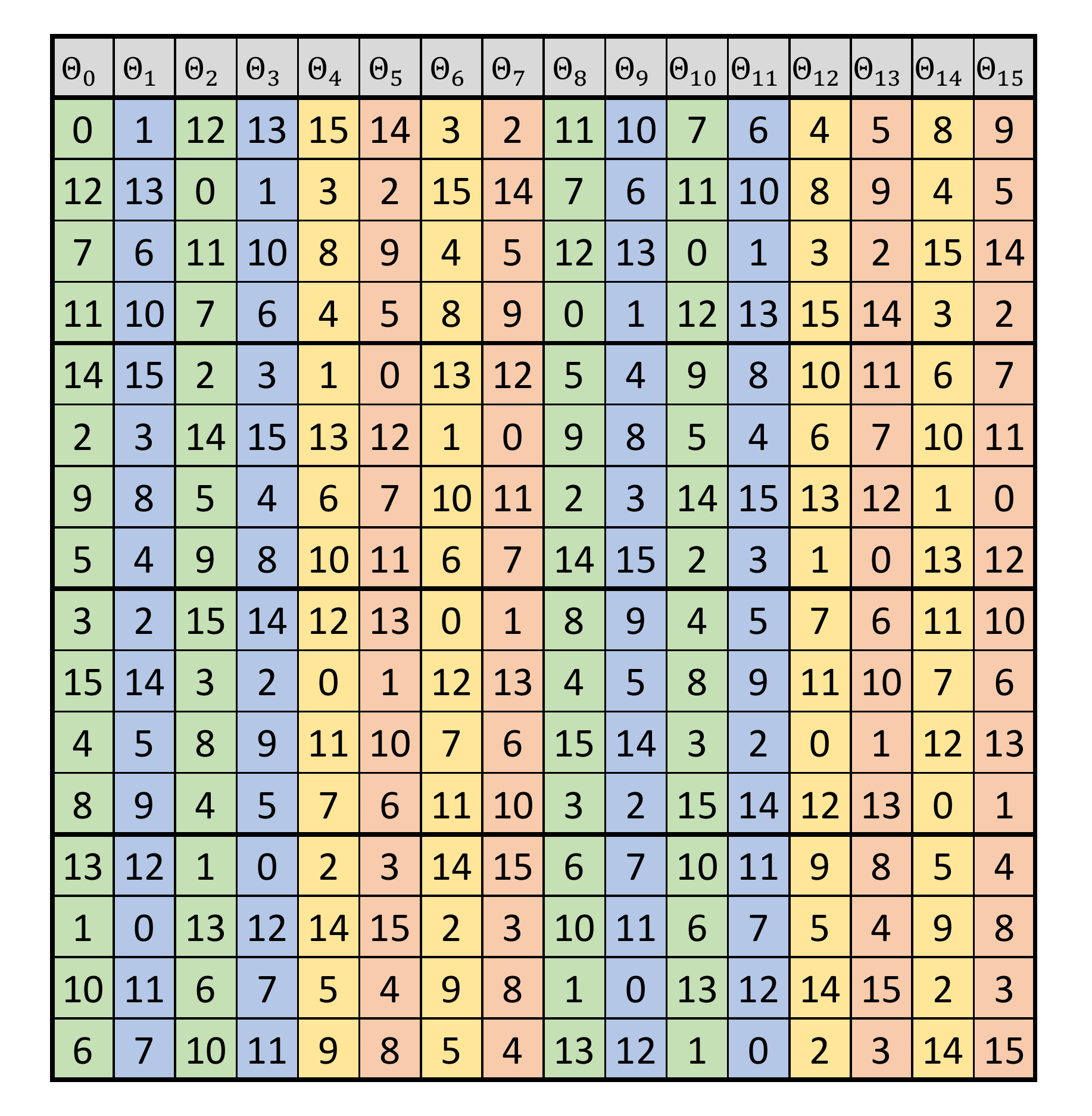} 
	\caption{all $\BOBMAT{\BUTI}$ matrices in the code of $\K=7$ and $polynomial=(171,133)$. This is not a tensor matrix! }
	\label{fig:BOMATs}
\end{figure}

\begin{figure}[t]
	\centering
	\includegraphics[width=0.9\columnwidth]{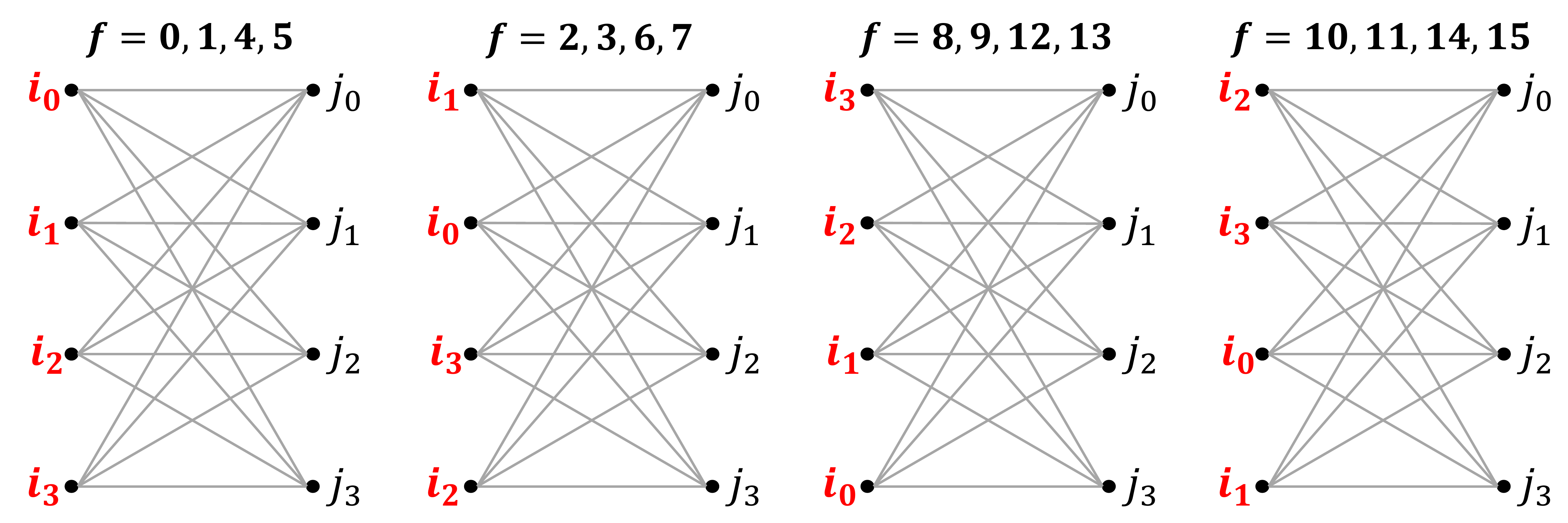} 
	\caption{The permutation in left states of dragonflies}
	\label{fig:permutation}
\end{figure}

%% file: results.tex
\section{Experimental Evaluation}

\subsection{Setup}

The proposed parallel algorithm is implemented in C language in the CUDA framework. The hardware platform is a server with an Intel Xeon CPU operating at $2.5$~GHz and Tesla V100 GPU. 
We employ Ubuntu OS $18.04$, gcc version $7.4$, and CUDA version $10.2$.

We experiment with a widely-used standard convolutional code, namely, $(2,1,7)$, i.e., code rate $\nicefrac{1}{2}$ and constraint length $7$, with generative polynomials $171$ and $133$. This configuration is shown in Fig.~\ref{fig:encoder}.

\subsection{BER Performance}

In order to verify the implementation, the system shown in Fig.~\ref{fig:verifSys} is employed. At the first step, a vector of uniformly distributed bits is generated and, at step $2$, passed to convolutional encoder. This part is the simulated transmitter. Then, at step $3$, encoded bits are transmitted in an AWGN channel with a specific $E_b/N_0$. Assuming that the BPSK modulation is used, the channel simulation is done by adding a vector of normally distributed values with standard deviation of $2^{\nicefrac{-(E_b/N_0)}{20}}$. Having generated a noisy coded vector, at step $4$, the simulated receiver can decode the signal and produce an output vector. At last, comparing the decoder output with the bits generated at the first step, Bit Error Rate (BER) will be obtained. It should be also noted that the BER value is reliable if enough data is generated and tested in the verification system. As a rule of thumb, if a vector of size $n$ is generated in the first step, only the BER value more than $\frac{100}{n}$ will be valid.

\begin{figure}[t]
	\centering
	\includegraphics[width=0.9\columnwidth]{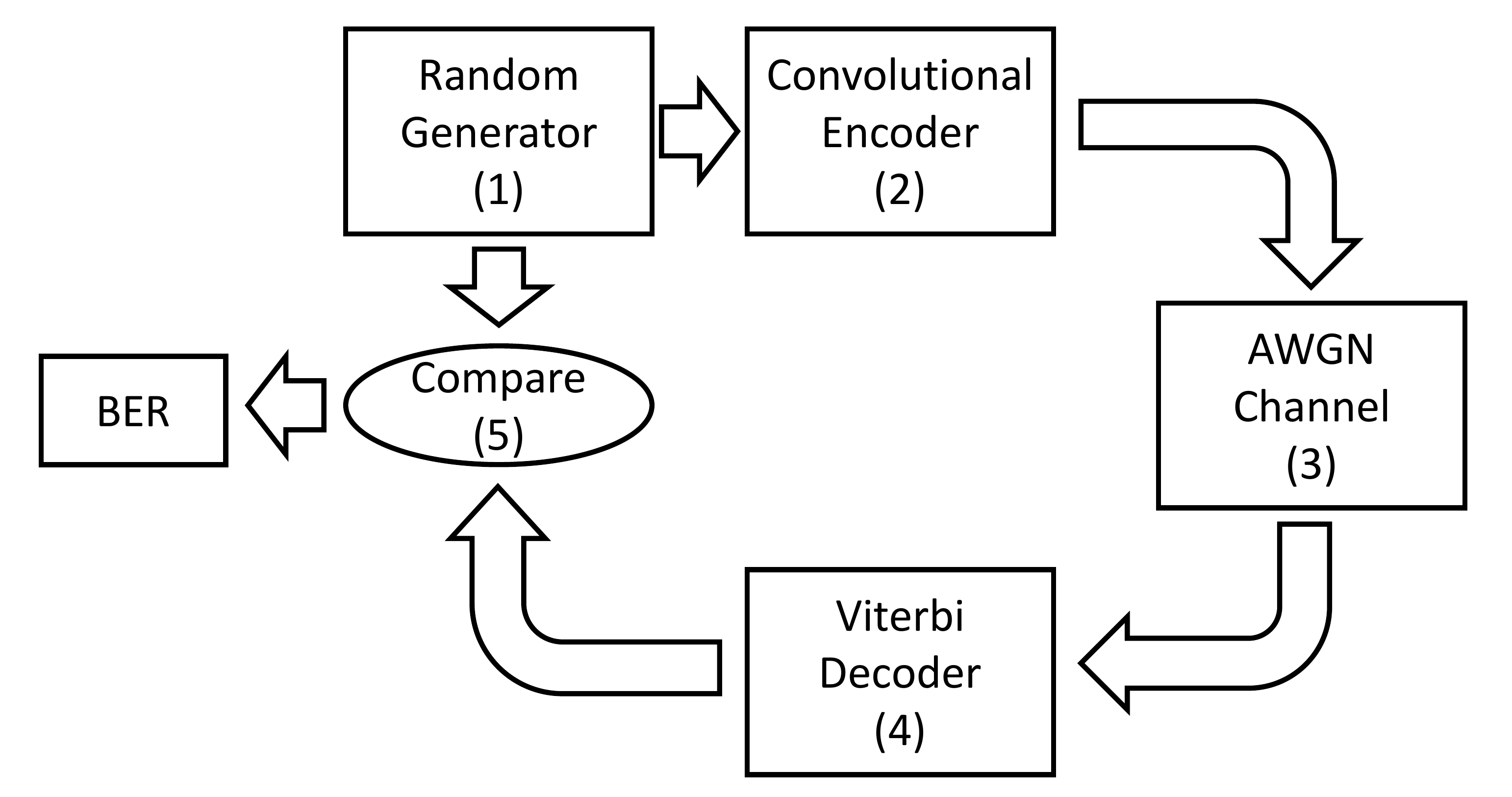} 
	\caption{The block diagram of the verification system}
	\label{fig:verifSys}
\end{figure}

\begin{figure}[t]
	\centering
	\includegraphics[width=0.9\columnwidth]{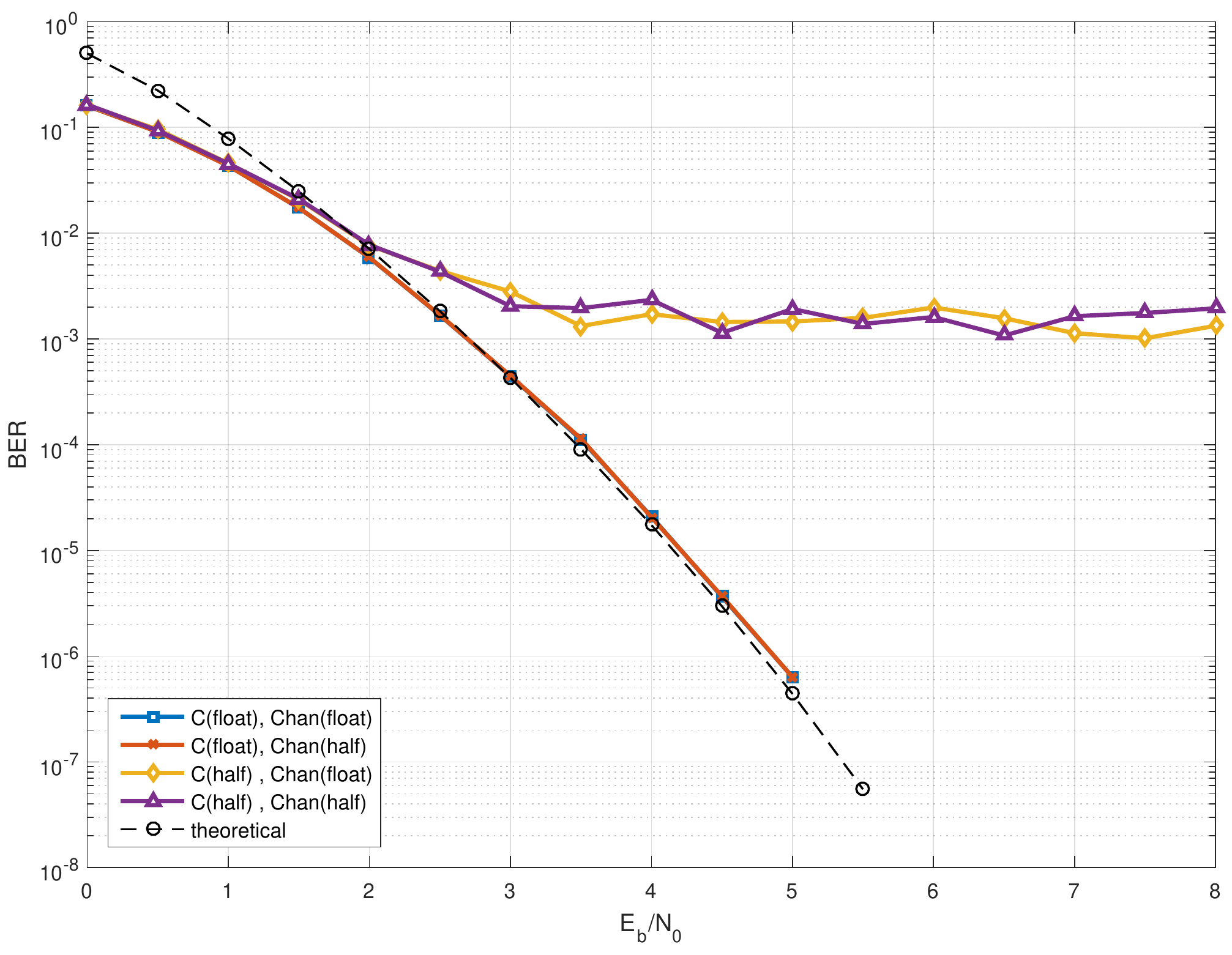} 
	\caption{Comparison between half-precision and single-precision operations in the Viterbi algorithm}
	\label{fig:ber:fh}
\end{figure}

The process shown in Fig.~\ref{fig:verifSys} can produce the BER for a specific $E_b/N_0$. Complete verification is done by drawing the BER curve over a range of $E_b/N_0$ values and comparing it with the theoretical one that can be generated by MATLAB BER tool which is invoked by "bertool" command. Implementation parameters can be tuned this way.

In this implementation, all calculations are done with floating-point variables. There are two types of floating-point variables, half-precision and single-precision. Size of a half-precision variable is $16$ bits which is half of that of a single-precision one. It means that memory size required by a half-precision array is less than a single-precision one, and consequently memory transfer time is decreased and throughput is increased. On the other hand, the result of single-precision operations are more reliable. Therefor, BER can be degraded by half-precision arrays. Thus, this selection is important and affects both throughput and BER. NVIDIA tensor cores provide the option of single-precision matrices only for $C$ and $D$. $A$ and $B$ have only half-precision option to store floating-point values.

In this paper, since $A$ and $B$ are not accumulated, their accuracy does not matter and half-precision arrays do not corrupt data. Nevertheless, $C$, which is also used as $D$, is accumulated along the trellis and is prune to inaccuracy and its preferred precision should be investigated. Another array that can be selected between single-precision and half-precision is the data received from channel. Since this data is going to be stored in $B$, it will be converted to half-precision. As a result, it can be half-precision in the first place.

Fig.~\ref{fig:ber:fh} shows the BER curve regarding different combinations of these two arrays.  Viterbi implementation results under both conditions. Obviously, half-precision tensor operations do not have enough accuracy to achieve satisfactory results. As expected, the data received from channel can be half-precision without any problem. Nonetheless, $C$ must be single-precision.

\subsection{Throughput}

Table.~\ref{table:throughput} shows throughput for different options precision. According to what discussed, $C$ must be single-precision. Therefore, the first two rows of the table are valid. It is also observed that, as expected, if channel is half-precision, throughput increases.

\begin{table}[tp]
\begin{center}
\caption{Decoder Throughput}
\label{table:throughput}
\begin{tabular}{|l|l||l|}
\hline
\textbf{C} & \textbf{channel} & \textbf{Throughput (Gb/s)} \\
\hline
\hline
single & single & $19.5$\\
\hline
single & half & $21.4$\\
\hline
half & single & $20.1$\\
\hline
half & half & $22.2$\\
\hline
\end{tabular}
\end{center}
\end{table}

\begin{figure*}[t]
	\centering
	\includegraphics[height=240mm]{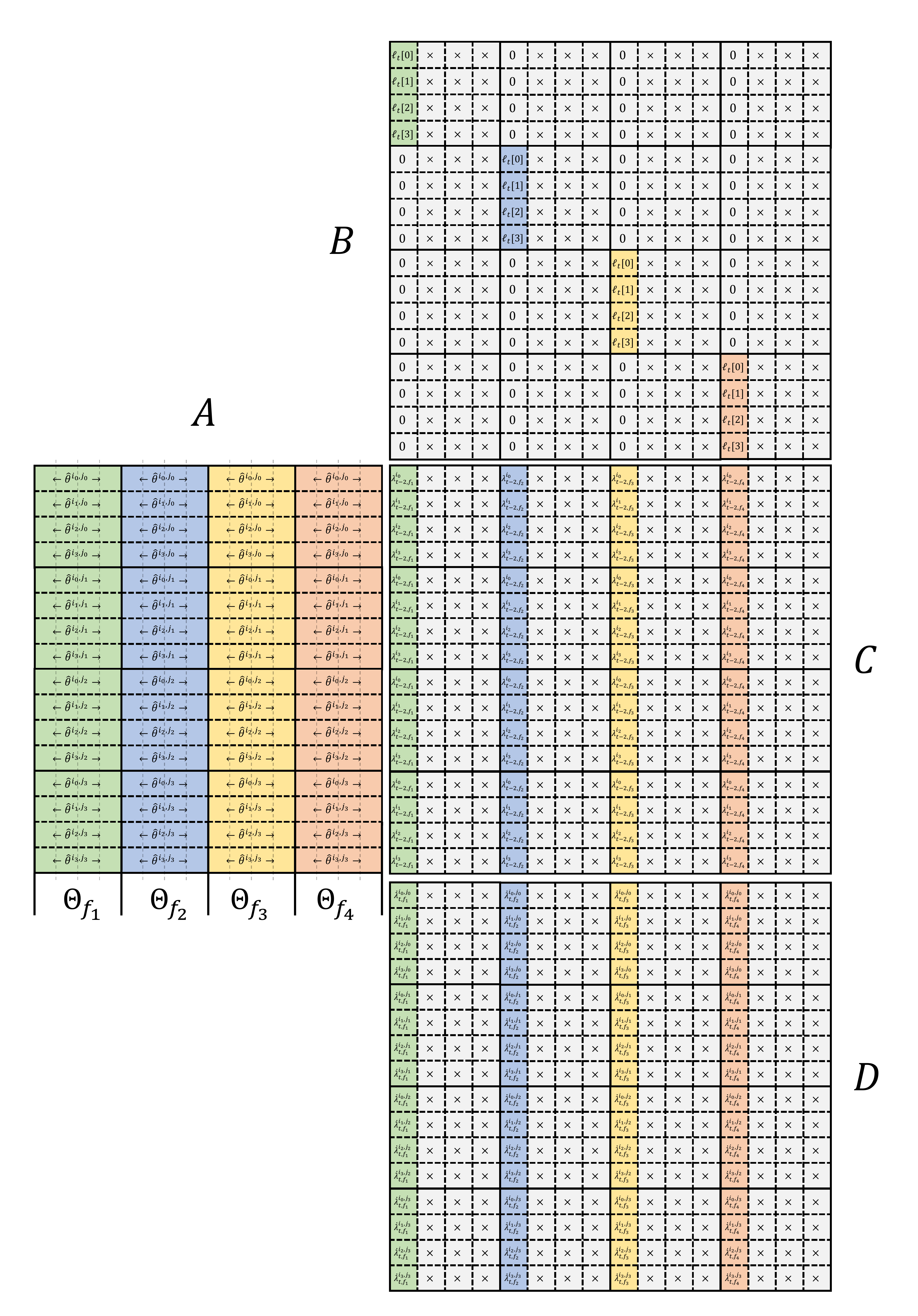} 
	\caption{Viterbi implementation on $16\times 16$ tensor cores based on radix-4 approach.}
	\label{fig:tensor16x16:simple}
\end{figure*}

\begin{figure*}[t]
	\centering
	\includegraphics[height=240mm]{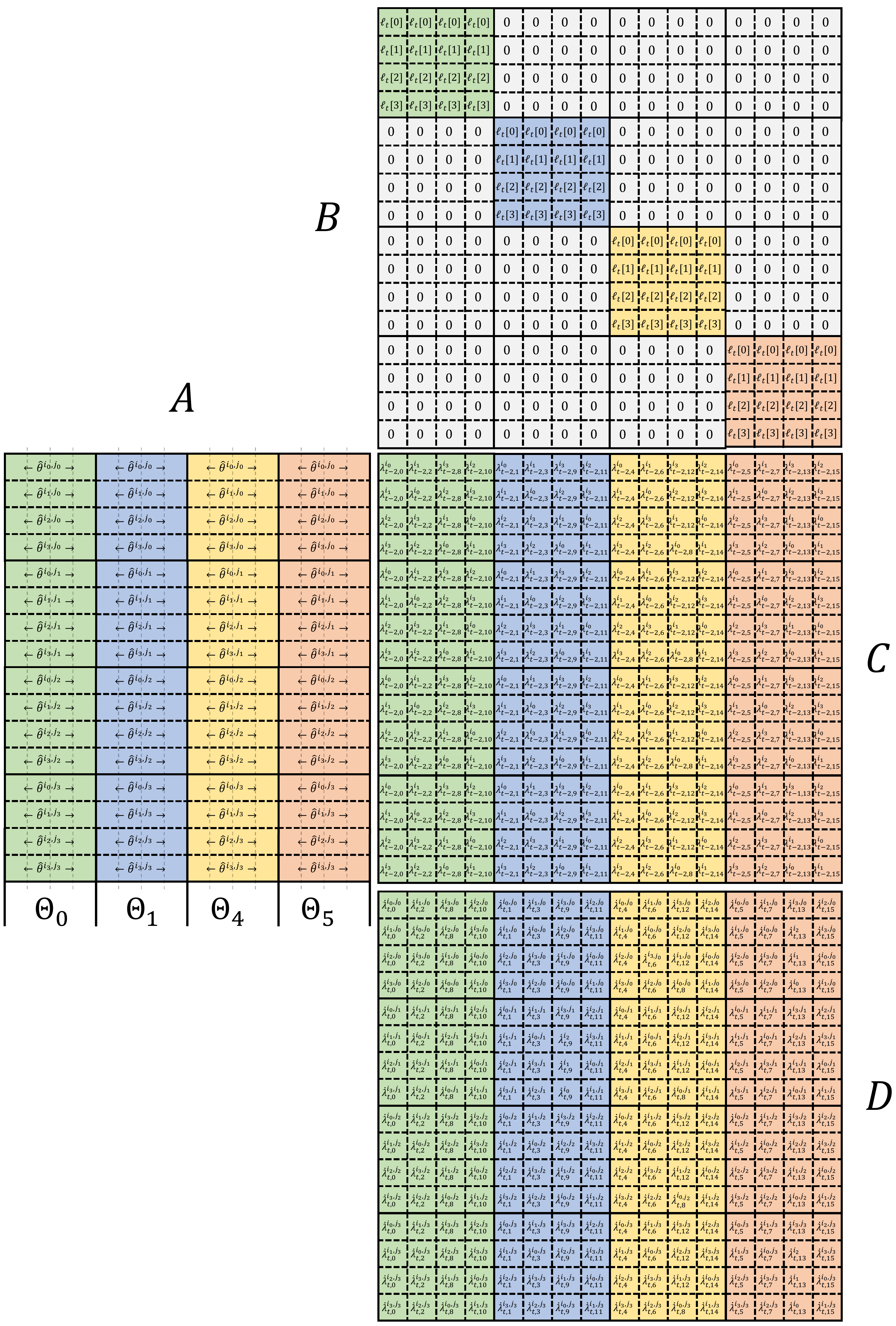} 
	\caption{Viterbi implementation of $\K=7$ and $polynomial=(171,133)$ on $16\times 16$ tensor cores based on radix-4 approach after efficient use of matrix entries.}
	\label{fig:tensor16x16:opt}
\end{figure*}